\documentclass[a4paper,UKenglish]{lipics-v2016}

\usepackage{microtype}
\usepackage{tikz}
\usepackage{tkz-graph}

\usepackage[numbers,sort]{natbib}

\newcommand{\bigO}{\mathcal{O}}

\title{String Attractors: Verification and Optimization}

\author[1]{Dominik Kempa}
\author[2]{Alberto Policriti}
\author[3]{Nicola Prezza}
\author[3]{Eva Rotenberg}
\affil[1]{University of Helsinki, Finland\\
	\texttt{dkempa@cs.helsinki.fi}}
\affil[2]{University of Udine, Udine, Italy\\
  \texttt{alberto.policriti@uniud.it }}
\affil[3]{Technical University of Denmark, Kgs. Lyngby, Denmark\\
  \texttt{\{npre,erot\}@dtu.dk}}
\authorrunning{D. Kempa, A. Policriti, N. Prezza, and E. Rotenberg} 

\Copyright{Dominik Kempa, Alberto Policriti, Nicola Prezza, and Eva Rotenberg}

\subjclass{Dummy classification -- please refer to \url{http://www.acm.org/about/class/ccs98-html}}

\EventEditors{Ioannis Chatzigiannakis, Christos Kaklamanis, Daniel Marx, and Don Sannella}
\EventNoEds{4}
\EventLongTitle{45th International Colloquium on Automata, Languages, and Programming (ICALP 2018)}
\EventShortTitle{ICALP 2018}
\EventAcronym{ICALP}
\EventYear{2018}
\EventDate{July 9--13, 2018}
\EventLocation{Prague, Czech Republic}
\EventLogo{eatcs}
\SeriesVolume{80}
\ArticleNo{}

\begin{document}

\maketitle

\begin{abstract}
	
	String attractors [STOC 2018] are combinatorial objects recently introduced to unify all known dictionary compression techniques in a single theory. A set $\Gamma\subseteq [1..n]$ is  a \emph{$k$-attractor} for a string $S\in[1..\sigma]^n$ if and only if every distinct substring of $S$ of length at most $k$ has an occurrence straddling at least one of the positions in $\Gamma$. Finding the smallest $k$-attractor is NP-hard for $k\geq3$, but polylogarithmic approximations can be found using reductions from dictionary compressors.
	It is easy to reduce the $k$-attractor problem to a set-cover instance where string's positions are interpreted as sets of substrings. The main result of this paper is a much more powerful reduction based on the truncated suffix tree.
	Our new characterization of the problem leads to more efficient algorithms for string attractors: we show how to check the validity and minimality of a $k$-attractor in near-optimal time and how to quickly compute exact and approximate solutions.
	For example, we prove that a minimum $3$-attractor  can be found in optimal $\bigO(n)$  time when $\sigma\in \bigO(\sqrt[3+\epsilon]{\log n})$ for any constant $\epsilon>0$, and $2.45$-approximation can be computed in  $\bigO(n)$ time on general alphabets.
	To conclude, we introduce and study the complexity of the closely-related \emph{sharp-$k$-attractor} problem: to find the smallest set of positions capturing all distinct substrings of length \emph{exactly} $k$. We show that the problem is in P for $k=1,2$ and is NP-complete for constant $k\geq 3$.
\end{abstract}

\section{Introduction}
The goal of dictionary compression is to reduce the size of an input string by exploiting its repetitiveness. In the last decades, several dictionary compression techniques---some more powerful than others---were developed to achieve this goal: Straight-Line programs~\cite{KY00} (context-free grammars generating the string), Macro schemes~\cite{storer1982data} (a set of substring equations having the string as unique solution), the run-length Burrows-Wheeler transform~\cite{burrows1994block} (a string permutation whose number of equal-letter runs decreases as the string's repetitiveness increases), and the compact directed acyclic word graph~\cite{blumer1987complete,crochemore1997direct} (the minimization of the suffix tree). Each scheme from this family comes with its own set of algorithms and data structures to perform compressed-computation operations---e.g. random access---on the compressed representation. Despite being apparently unrelated, in~\cite{kempa2017roots} all these compression schemes were proven to fall under a common general scheme: they all induce a set $\Gamma\subseteq  [1..n]$ whose cardinality is bounded by the compressed representation's size and with the property that each distinct substring has an occurrence crossing at least a position in $\Gamma$. A set with this property is called a \emph{string attractor}. 
Intuitively, positions in a string attractor  capture ``interesting'' regions of the string; a string of low complexity (that is, more compressible), will generate a smaller attractor. 
Surprisingly, given such a set one can build a data structure of size $\bigO(|\Gamma|\ \mathtt{polylog }(n))$ supporting random access queries in optimal time~\cite{kempa2017roots}: string attractors therefore provide a universal framework for performing compressed computation on top of \emph{any} dictionary compressor (and even optimally for particular queries such as random access). 

These premises suggest that an algorithm  computing a smallest string attractor for a given string would be an invaluable tool for designing better compressed data structures.
Unfortunately, computing a minimum string attractor is NP-hard. The problem remains NP-hard even under the restriction that only substrings of length at most $k$ are captured by $\Gamma$, for any $k\geq 3$ and on large alphabets. In this case, we refer to the problem as \emph{$k$-attractor}. Not all hope is lost, however: as shown in~\cite{kempa2017roots}, dictionary compressors are actually heuristics for computing a small $n$-attractor (with polylogarithmic approximation rate w.r.t. the minimum), and, more in general, $k$-attractor admits a $\bigO(\log k)$-approximation based on a reduction to set cover. It is actually easy to find such a reduction: choose as universe the set of distinct substrings and as set collection the string's positions (i.e. set $s_i$ contains all substrings straddling position $i$). The main limitation of this approach is that the set of distinct substrings could be quadratic in size; this makes the strategy of little use in cases where the goal is to design usable (i.e. as close as possible to linear-time) algorithms on string attractors.

The core result of this paper is a much more powerful reduction from $k$-attractor to set-cover: the universe $\mathcal U$ of our instance is equal to the set of edges of the $k$-truncated suffix tree, while the size of the set collection $\mathcal S\subseteq 2^{\mathcal U}$ is bounded by the size of the $(2k-1)$-truncated suffix tree. First of all, we obtain a universe that is always at least $k$ times smaller than the naive approach. Moreover, the size of our set-cover instance does not depend on the string's length $n$, unless $\sigma$ and $k$ do. This allows us to show that $k$-attractor is actually solvable in polynomial time for small values of $k$ and $\sigma$, and leads us to efficient algorithms for a wide range of different problems on string attractors.

The paper is organized as follows. In Section \ref{sec:notation} we describe the notation used throughout the paper and we report the main notions related to $k$-attractors. In Section \ref{sec:results} we give the main theorem stating our reduction to set-cover (Theorem \ref{th:main}) and briefly discuss the results that we obtain in the rest of the paper by applying it. Theorem \ref{th:main} itself is proven in Section \ref{sec:characterizations}, together with other useful lemmas that will be used  to further improve the running times of the algorithms described in Section \ref{sec: faster algorithms}. Finally, in Section \ref{section:sharp} we introduce and study the complexity of the closely-related \emph{sharp-$k$-attractor} problem: to capture all distinct substrings of length exactly $k$.

All proofs omitted for space reasons can be found in the Appendix.



\subsection{Notation and definitions}\label{sec:notation}

$[i..j]$ indicates the set $\{i, i+1, \dots j\}$. The notation $S[1,n]$ indicates a string of length $n$ with indices starting from $1$. 
When $D$ is an array (e.g. a string or an integer array), $D[i..j]$ indicates the sub-array $D[i], D[i+1], \dots, D[j]$.

We assume the reader to be familiar with the notions of \emph{suffix tree}~\cite{weiner1973linear}, \emph{suffix array}~\cite{manber1993suffix}, and \emph{wavelet tree}~\cite{grossi2003high,navarro2014wavelet}.
$ST^k(S)$ denotes the $k$-truncated suffix tree of $S$, i.e. the compact trie containing all substrings of $S$ of length at most $k$.  $\mathcal E(\mathcal T)$ denotes the set of edges of the compact trie $\mathcal T$.  $\mathcal L(\mathcal T)$ denotes the set of leaves at maximum string depth of the compact trie $\mathcal T$ (i.e. leaves whose string depth is equal to the maximum string depth among all leaves). 
Let $e=\langle u,v \rangle$ be an edge in the (truncated) suffix tree of $S$. With $s(e)$ we denote the string read from the suffix tree root to the first character in the label of $e$. $\lambda(e) = |s(e)|$ is the length of this string. We will also refer to $\lambda(e)$ as the \emph{string depth} of $e$.
Note that edges $e_1,\dots, e_t$ of the $k$-truncated suffix tree have precisely the same labels $s(e_1), \dots, s(e_t)$ of the suffix tree edges $e'_1,\dots, e'_t$ at string depth  $\lambda(e'_i) \leq k$. It follows that we can use these two edge sets interchangeably when we are only interested in their labels (this will be the case in our results). 
Let $SA[1,n]$ denote the suffix array of $S$.
$\langle l_e,r_e \rangle$, with $e\in \mathcal E(ST^k(S))$  being an edge in the $k$-truncated suffix tree, will denote the suffix array range corresponding to the string $s(e)$, i.e. $SA[l_e..r_e]$ contains all suffixes prefixed by $s(e)$. 

Unless otherwise specified, we give the space of our data structures in words of $\Theta(\log n)$ bits each.

With the following definition we recall the notion of \emph{$k$-attractor} of a string~\cite{kempa2017roots}.

\begin{definition}\label{def: string k-attractor}
	A \emph{$k$-attractor} of a string $S\in\Sigma^n$ is a set of positions $\Gamma \subseteq [1..n]$ such that every substring $S[i..j]$ with $i \leq j < i+k$ has at least one occurrence $S[i'..j'] = S[i..j]$ with $j'' \in [i'..j']$ for some $j''\in\Gamma$. 
\end{definition}

We call \emph{attractor} a $n$-attractor for $S$.

\begin{definition}\label{def: minimal k-attractor}
	A \emph{minimal $k$-attractor} of a string $S\in\Sigma^n$ is a $k$-attractor $\Gamma$ such that $\Gamma - \{j\}$ is not a $k$-attractor of $S$ for any $j\in\Gamma$.
\end{definition}

\begin{definition}\label{def: minimum k-attractor}
	A \emph{minimum $k$-attractor} of a string $S\in\Sigma^n$ is a $k$-attractor $\Gamma^*$ such that, for any $k$-attractor $\Gamma$ of $S$, $|\Gamma^*| \leq |\Gamma|$.
\end{definition}

\begin{theorem}~\cite[Thm. 8]{kempa2017roots}
	The problem of deciding whether a string $S$ admits a $k$-attractor of size at most $t$ is NP-complete for $k\geq 3$.
\end{theorem}

Note that we can pre-process the input string $S$ so that its characters are mapped into the range $[1..n]$. This transformation can be computed in linear time and space. It is easy to see that a set $\Gamma$ is a $k$-attractor for $S$ if and only if it is a $k$-attractor for the transformed string. It follows that we do not need to put any restriction on the alphabet of the input string, and in the rest of the paper we can assume that the alphabet is $\Sigma = [1..\sigma]$, with $\sigma\leq n$.

\subsection{Overview of the contributions}\label{sec:results}

Our main theorem is a reduction to set-cover based on the notion of truncated suffix tree:

\begin{theorem}\label{th:main}
	Let $S'=\#^{k-1}S\#^{k-1}$, with $\#\notin \Sigma$.
	$k$-attractor can be reduced to a set-cover instance 
	with universe $\mathcal U = \mathcal E(ST^k(S))$ and set collection $\mathcal S$ of size $|\mathcal L(ST^{2k-1}(S'))|$. 
\end{theorem}

Figure \ref{fig:st} depicts the main technique (Lemma \ref{theorem:k-attr characterization}) standing at the core of our reduction: a set $\Gamma$ is a valid attractor if and only if it marks (or colours, in the picture), all suffix tree edges.

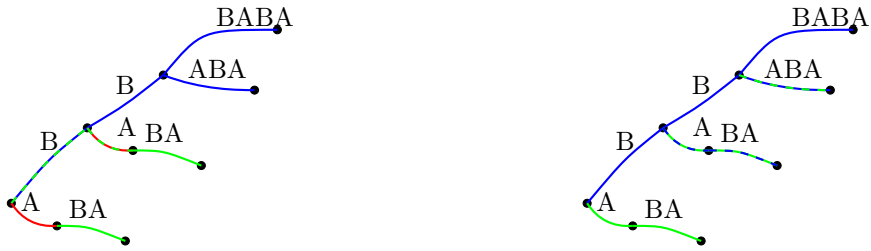
\begin{figure}[h]
	\begin{subfigure}{.48\textwidth}
		\begin{tikzpicture}
		\tikzstyle{vertex}=[circle,fill, draw=black,minimum size=3pt,inner sep=0pt]
		\tikzstyle{leaf}=[circle,fill, draw=black,minimum size=3pt,inner sep=0pt]
		
		\node[vertex] (r) at (0,0) {};
		\node[vertex] (b1) at (1,1) {};
		\node[vertex] (a1) at (0.6,-0.3) {};
		\node[leaf] (l1) at (1.5,-0.5) {};
		\node[vertex] (b2) at (2,1.7) {};
		\node[vertex] (a2) at (1.6,0.7) {};
		\node[leaf] (l2) at (2.5,0.5) {};
		\node[leaf] (l3) at (3.2,1.5) {};
		\node[leaf] (l4) at (3.5,2.3) {};

		\draw [blue,thick] (0,0) .. controls (0.5,0.6)  .. (1,1) node[midway, above, black] {B};
		\draw [green,thick,dashed] (0,0) .. controls (0.5,0.6)  .. (1,1) node[midway, above, black]{};
		\path [red,thick] (0,0) edge [bend right] node[above,black] {A} (0.6,-0.3); 
		\draw [green,thick] (0.6,-0.3) .. controls (1,-0.3) .. (1.5,-0.5) node[midway, above, black] {BA};
		\draw [blue,thick] (1,1) .. controls (1.5,1.3)  .. (2,1.7) node[midway, above, black] {B};
		\path [red,thick] (1,1) edge [bend right] node[above right,black] {A} (1.6,0.7); 
		\path [green,thick, dashed] (1,1) edge [bend right] (1.6,0.7); 
		\draw [green,thick] (1.6,0.7) .. controls (2,0.7) .. (2.5,0.5) node[midway, above, black] {BA};
		\draw [blue,thick] (2,1.7) .. controls (2.5,2.3)  .. (3.5,2.3) node[midway, above right, black] {BABA};
		\draw [blue,thick] (2,1.7) .. controls (2.5,1.5) and (3,1.5) .. (3.2,1.5) node[midway, above, black] {ABA};
		\end{tikzpicture}
		\caption{The suffix tree of the string B\textcolor{blue}{B}BA\textcolor{green}{B}\textcolor{red}{A}, coloured corresponding to attractor positions $\textcolor{blue}{2},\textcolor{green}{5},\textcolor{red}{6}$. This attractor is minimal: removing any position leaves some edge uncoloured.}
	\end{subfigure}	
	\hspace{0.05\textwidth}
	\begin{subfigure}{.47\textwidth}
		\begin{tikzpicture}
		\tikzstyle{vertex}=[circle,fill, draw=black,minimum size=3pt,inner sep=0pt]
		\tikzstyle{leaf}=[circle,fill, draw=black,minimum size=3pt,inner sep=0pt]
		
		\node[vertex] (r) at (0,0) {};
		\node[vertex] (b1) at (1,1) {};
		\node[vertex] (a1) at (0.6,-0.3) {};
		\node[leaf] (l1) at (1.5,-0.5) {};
		\node[vertex] (b2) at (2,1.7) {};
		\node[vertex] (a2) at (1.6,0.7) {};
		\node[leaf] (l2) at (2.5,0.5) {};
		\node[leaf] (l3) at (3.2,1.5) {};
		\node[leaf] (l4) at (3.5,2.3) {};
		
		\draw [blue,thick] (0,0) .. controls (0.5,0.6)  .. (1,1) node[midway, above, black] {B};
		\path [green,thick] (0,0) edge [bend right] node[above,black] {A} (0.6,-0.3); 
		\draw [green,thick] (0.6,-0.3) .. controls (1,-0.3) .. (1.5,-0.5) node[midway, above, black] {BA};
		
		\draw [blue,thick] (1,1) .. controls (1.5,1.3)  .. (2,1.7) node[midway, above, black] {B};
		\path [green,thick] (1,1) edge [bend right] node[above right,black] {A} (1.6,0.7); 
		\path [blue,thick, dashed] (1,1) edge [bend right] (1.6,0.7); 
		\draw [green,thick] (1.6,0.7) .. controls (2,0.7) .. (2.5,0.5) node[midway, above, black] {BA};
		\draw [blue,dashed,thick] (1.6,0.7) .. controls (2,0.7) .. (2.5,0.5);
		\draw [blue,thick] (2,1.7) .. controls (2.5,2.3)  .. (3.5,2.3) node[midway, above right, black] {BABA};
		\draw [blue,thick] (2,1.7) .. controls (2.5,1.5) and (3,1.5) .. (3.2,1.5) node[midway, above, black] {ABA};
		\draw [green,dashed,thick] (2,1.7) .. controls (2.5,1.5) and (3,1.5) .. (3.2,1.5);
		\end{tikzpicture}
		\caption{The suffix tree of the string BB\textcolor{blue}{B}\textcolor{green}{A}BA, coloured corresponding to attractor positions $\textcolor{blue}{3},\textcolor{green}{4}$. This attractor is minimum: it is minimal and of minimum size.}
	\end{subfigure}
	\caption{
		A position $i$ ``marks'' (or, here, \emph{colours}) a suffix tree edge $e$ if and only if it crosses an occurrence of the string read from the root to the first letter in the label of $e$.
		A set of positions forms a $k$-attractor if and only if they colour all edges of the $k$-truncated suffix tree (in this figure, $k=6$ and we colour the whole suffix tree).
		Dashed lines indicate that the edge has multiple colours.
		The string terminator $\$$ (and edges labeled with $\$$) is omitted for simplicity.}\label{fig:st}
\end{figure}

Using the reduction of Theorem \ref{th:main}, we obtain the following results. 
First, we present efficient algorithms to check the validity and minimality of a $k$-attractor. Note that it is trivial to perform these checks in cubic time (or quadratic with little more effort). In Theorem \ref{th:checking} we show that we can check whether a set $\Gamma\subseteq [1..n]$ is a valid $k$-attractor for $S$ in $\bigO(n)$ time and  $\bigO(n)$ words of space. Using recent advances in compact data structures, we show how to further reduce this working space to $\bigO(n\log\sigma)$ bits without affecting query times when $k\leq \sigma^{\bigO(1)}$, or with a small time penalty in the general case. In particular, when $k$ is constant we can always check the correctness of a $k$-attractor in $\bigO(n)$ time and $\bigO(n\log\sigma)$ bits of space.
With similar techniques, in Theorem \ref{theorem:minimality check} we show how to verify that  $\Gamma$ is a \emph{minimal} $k$-attractor for $S$ in near-optimal $\bigO(n\log n)$ time. 

In Theorem \ref{th:reporting} we  show that the  structure used in Theorems \ref{th:checking} and \ref{theorem:minimality check} can be augmented (using a recent result on weighted ancestors on suffix trees~\cite{gawrychowski2014weighted}) to support reporting all occurrences of a substring straddling a position in $\Gamma$ in optimal time.

We then focus on optimization problems. In Theorem \ref{th:minimum} we show that a minimum $k$-attractor
can be found in $\bigO(n) + exp\big(\bigO(\sigma^k\log\sigma^k)\big)$ time. Similarly, in Theorem \ref{theorem:find minimal} we show that a minimal $k$-attractor
can be found in $\bigO(\min\{n^2,nk^2, n+\sigma^{2k-1} k^2\})$ expected time. 
With Theorem \ref{theorem:minimal approx rate} we show that minimal $k$-attractors are within a  factor of $k$ from the optimum, therefore proving that Theorem \ref{theorem:find minimal} actually yields an approximation algorithm.
In Theorem \ref{theorem:approx} we show that within the same time  of Theorem \ref{theorem:find minimal} we can compute an exponentially-better $\bigO(\log k)$-approximation.

Theorems \ref{th:minimum}, \ref{theorem:find minimal}, and  \ref{theorem:approx}  yield the following corollaries:

\begin{corollary}\label{cor:k-attr P}
	$k$-attractor is in P when $\sigma^k \log\sigma^k \in \bigO(\log n)$.
\end{corollary}

\begin{corollary}
	For constant $k$, a minimum $k$-attractor can be found in optimal $\bigO(n)$ time when $\sigma^{k+\epsilon} \in \bigO(\log n)$, for any constant $\epsilon>0$.
\end{corollary}
\begin{proof}
	Pick any constant $\epsilon>0$. Then, $\sigma^{k+\epsilon} = \sigma^{k(1+\epsilon')}$, where $\epsilon'=\epsilon/k > 0$ is a constant. On the other hand, for any constant $\epsilon'>0$ we have that $\sigma^k \log\sigma^k \in o(\sigma^k\cdot(\sigma^k)^{\epsilon'}) = o(\sigma^{k(1+\epsilon')})$. It follows that $\sigma^k \log\sigma^k \in o(\log n)$, i.e. by Theorem \ref{th:minimum} we can find a minimum $k$-attractor in linear time. 
\end{proof}

With our new results we can, for example (keep in mind that $3$-attractor  is NP-complete for general alphabets): 

\begin{itemize}
	\item Find a minimum $3$-attractor in $\bigO(n)$ time when $\sigma\in \bigO(\sqrt[3+\epsilon]{\log n})$, for any $\epsilon>0$.
	\item Find a minimal $3$-attractor in $\bigO(n)$ expected time.  This, by Theorem \ref{theorem:minimal approx rate}, is a $3$-approximation to the minimum.
	\item Find a $2.45$-approximation of the minimum $3$-attractor in $\bigO(n)$ expected time.
\end{itemize}

To conclude, in Section \ref{section:sharp} we study the \emph{sharp-$k$-attractor} problem: to find  a smallest set of positions covering all substrings of length exactly $k$. In Theorem \ref{th:sharp-attractor-npc} we show that the problem is NP-complete for constant $k\geq 3$, while in Theorem \ref{thm:2-sharp-attractor} we give a polynomial-time algorithm for the case $k=2$.

\section{A better reduction to set-cover}\label{sec:characterizations}

In this section we give our main result: a smaller reduction from $k$-attractor to set-cover. We start with an alternative characterization of $k$-attractors based on the $k$-truncated suffix tree.

\begin{definition}[Marker]\label{def:marker}
	$j\in\Gamma$ is a \emph{marker} for a suffix tree edge $e$ if and only if 
	$$\exists i\in SA[l_e..r_e]\ :\ i \leq j < i+\lambda(e)$$
	Equivalently, we say that $j$ \emph{marks} $e$.
\end{definition}

\begin{definition}[Edge marking]\label{def:edge marking}
	$\Gamma \subseteq [1..n] $ \emph{marks} a suffix tree edge $e$ if and only if there exists a $j\in \Gamma$ that marks $e$.
\end{definition}

\begin{definition}[Suffix tree $k$-marking]
	$\Gamma \subseteq [1..n]$ is a \emph{suffix tree $k$-marking} if and only if it marks every edge $e$ such that $\lambda(e)\leq k$ (equivalently, every $e\in\mathcal E(ST^k(S))$). 
\end{definition}

When $k=n$ we simply say \emph{suffix tree marking} (since all edges satisfy $\lambda(e)\leq n$).
We now show that the notions of $k$-attractor and suffix tree $k$-marking are equivalent.

\begin{lemma}\label{theorem:k-attr characterization}
	$\Gamma$ is a $k$-attractor if and only if it is a suffix tree $k$-marking.
\end{lemma}	
\begin{proof}
	$(\Rightarrow)$ Let $\Gamma$ be a $k$-attractor. Pick any suffix tree edge $e$ such that $\lambda(e)\leq k$. Then, $\lambda(e) = |s(e)| \leq k$ and, by definition of $k$-attractor, there exists a $j\in \Gamma$ and a $i$ such that $s(e) = S[i..i+|s(e)|-1]$ and $i \leq j \leq i+|s(e)|-1$. We also have that $i \in SA[l_e..r_e]$ (being $\langle l_e,r_e\rangle$ precisely the suffix array range of suffixes prefixed by $s(e)$). Putting these results together, we found a $i \in SA[l_e..r_e]$ such that $i \leq j \leq i+\lambda(e)-1$ for some $j\in \Gamma$, which by Definition \ref{def:edge marking} means that $\Gamma$ marks $e$. Since the argument works for any edge $e$ at string depth at most $k$, we obtain that $\Gamma$ is a suffix tree $k$-marking.
	
	$(\Leftarrow)$ Let $\Gamma$ be a suffix tree $k$-marking. Let moreover $s$ be a substring of $S$ of length at most $k$. Consider the lowest suffix tree edge $e$ (i.e. with maximum $\lambda(e)$) such that $s(e)$ prefixes $s$. In particular, $\lambda(e)\leq k$.
	Note that, by definition of suffix tree, every occurrence $S[i..i+|s(e)|-1] = s(e)$ of $s(e)$ in $S$ prefixes an occurrence of $s$: $S[i..i+|s|-1] = s$. By definition of $k$-marking, there exists a $j\in \Gamma$ such that $j$ is a marker for $e$, which  means (by Definition \ref{def:marker}) that $\exists i\in SA[l_e..r_e]\ :\ i \leq j < i+\lambda(e)$. Since $i\in SA[l_e..r_e]$, $SA[i]$ is an occurrence of $s(e)$, and therefore of $s$. But then, we have that $i \leq j < i+\lambda(e) = i+|s(e)| \leq i+|s|$, i.e. $S[SA[i]..SA[i]+|s|-1]$ is an occurrence of $s$ straddling $j\in \Gamma$. Since the argument works for every substring $s$ of $S$ of length at most $k$, we obtain that $\Gamma$ is a $k$-attractor.	
\end{proof}

An equivalent formulation of Lemma \ref{theorem:k-attr characterization} is that $\Gamma$ is a $k$-attractor if and only if it marks all edges of the $k$-truncated suffix tree.
In particular (case $k=n$),  $\Gamma$ is an attractor if and only if it is a suffix tree marking.

Lemma \ref{theorem:k-attr characterization} will be used to obtain a smaller universe $\mathcal U$ in our set-cover reduction.
With the following Lemmas we show that also the size of the set collection $\mathcal S$ can be considerably reduced when $k$ and $\sigma$ are small.

\begin{definition}[$k$-equivalence]
	Two positions $i,j\in [1..n]$ are $k$-equivalent, indicated as $i\equiv _k j$, if and only if
	$$
	S'[ i-k+1..i+k-1 ] = S'[ j-k+1..j+k-1 ]
	$$
	where $S'[i] = \#$ if $i<1$ or $i>n$ (note that we allow negative positions) and $S'[i] = S[i]$ otherwise, and $\#\notin \Sigma$ is a new character.
\end{definition}

It is easy to see that $k$-equivalence is an equivalence relation. First, we bound the size of the distinct equivalence classes of $\equiv _k$ (i.e. the size of the quotient set $[1..n]/\equiv _k$).

\begin{lemma}\label{lemma:size}
	$|[1..n]/\equiv _k| = |\mathcal L(ST^{2k-1}(S'))| \leq \min\{n,\sigma^{2k-1} + 2k-2\}$
\end{lemma}

We now show that any minimal $k$-attractor can have at most one element from each equivalence class of $\equiv_k$. 

\begin{lemma}\label{lemma:equiv_classes}
	If $\Gamma$ is a minimal $k$-attractor, then for any $1\leq i \leq n$ it holds $|\Gamma \cap [i]_{\equiv_k}| \leq 1$.
\end{lemma}

Moreover, if we swap any element of a $k$-attractor with an equivalent element then the resulting set is still a $k$-attractor:

\begin{lemma}\label{lemma:swap}
	Let $\Gamma$ be a $k$-attractor. Then, $(\Gamma - \{j\}) \cup \{j'\}$ is a $k$-attractor for any $j\in \Gamma$ and any $j'\equiv_k j$.
\end{lemma}
\begin{proof}
	Pick any occurrence of a substring $s$, $|s|\leq k$, straddling position $j$. By definition of $\equiv_k$, since $j'\equiv_k j$ there is also an occurrence of $s$ straddling $j'$. This implies that  $\Gamma' = (\Gamma - \{j\}) \cup \{j'\}$ is a $k$-attractor. 
\end{proof}

Lemmas \ref{lemma:equiv_classes} and \ref{lemma:swap} imply that we can reduce the set of candidate positions from $[1..n]$ to $\mathcal C = \{\min(I)\ :\ I\in [1..n]/\equiv _k\}$ (that is, an arbitrary representative---in this case, the minimum---from any class of $\equiv _k$), and still be able to find a minimal/minimum $k$-attractor. Note that, by Lemma \ref{lemma:size}, $|[1..n]/\equiv _k| \leq \min\{n,\sigma^{2k-1} + 2k-2\}$.

We can now prove our main theorem.

\begin{proof}[Proof of Theorem \ref{th:main}]
	We build our set-cover instance $\langle \mathcal U, \mathcal S\rangle$ as follows. We choose $\mathcal U = \mathcal E(ST^k(S))$, i.e. the set of edges of the $k$-truncated suffix tree. The set collection $\mathcal S$ is defined as follows. Let $s_i = \{e\in\mathcal E(ST^k(S)) \ :\ i\ marks\ e\}$ and $\mathcal C = \{\min(I)\ :\ I\in [1..n]/\equiv _k\}$. Then,  we choose
	$$
	\mathcal S = \{s_i\ :\ i\in \mathcal C\}
	$$
	By the way we defined $\equiv _k$, each $I\in [1..n]/\equiv _k$ is unambiguously identified by a substring of length $2k-1$ of the string $S' = \#^{k-1}S\#^{k-1}$. 
	We therefore obtain $|\mathcal S| = |\mathcal L(ST^{2k-1}(S'))|$. We now prove correctness and completeness of the reduction. 
	
	\textbf{Correctness} By the definition of our reduction, a solution $\{s_{i_1} ,\dots, s_{i_\gamma}\}$ to $\langle \mathcal U, \mathcal S\rangle$ yields a set $\Gamma = \{i_1 ,\dots, i_\gamma\}$ of positions marking every edge in $\mathcal E(ST^k(S))$. 
	Then, Lemma \ref{theorem:k-attr characterization} implies that  $\Gamma$ is a $k$-attractor.
	
	\textbf{Completeness} Let $\Gamma = \{i_1 ,\dots, i_\gamma\}$ be a minimal $k$-attractor. Then, Lemmas \ref{lemma:equiv_classes} and \ref{lemma:swap} imply that the following set is also a minimal $k$-attractor of the same size:  $\Gamma' = \{j_1=min([i_1]_{\equiv_k}) ,\dots, j_\gamma = min([i_\gamma]_{\equiv_k})\}$. 
	Note that $\Gamma' \subseteq \{\min(I)\ :\ I\in [1..n]/\equiv _k\}$.
	By Lemma \ref{theorem:k-attr characterization}, $\Gamma'$ marks every edge in $\mathcal E(ST^k(S))$. Then, by definition of our reduction the collection $\{s_{j_1}, \dots, s_{j_\gamma}\}$ covers $\mathcal U = \mathcal E(ST^k(S))$.
	
\end{proof}

In the rest of the paper, we use the notation $\mathcal U = \mathcal E(ST^k(S))$ and $\mathcal C = \{\min(I)\ :\ I\in [1..n]/\equiv _k\}$ to denote the universe to be covered (edges of the $k$-truncated suffix tree) and the candidate attractor positions, respectively. Recall moreover that $|\mathcal U| \leq  \min\{n,\sigma^k\}$ and $|\mathcal C| \leq \min\{n,\sigma^{2k-1} + 2k-2\}$.

\subsection{Marker graph}

In this section we  introduce a graph that will play a crucial role in our approximation algorithms of Sections \ref{sec:minimal} and \ref{sec:approx}: our  algorithms will take a time linear in the graph's size to compute a solution. Intuitively, this graph represents the inclusion relations of the set-cover instance of Theorem \ref{th:main}.

\begin{definition}[Marker graph]\label{def:marking graph}
	
	Given a positive integer $k$, the \emph{marker graph} $\mathcal G_{S,k}$ of string $S$ is a bipartite undirected graph $\mathcal G_{S,k} = \langle \mathcal C, \mathcal U, E\rangle$, where the set $E \subseteq \mathcal C \times \mathcal U$ of edges is defined as
	$$
	E = \{\langle j,e \rangle\ :\ j\ marks\ e\}
	$$	
\end{definition}

\begin{lemma}\label{lemma:size of G}
	$|\mathcal G_{S,k}| \in \bigO(|E|) \subseteq \bigO(|\mathcal C|\cdot \min\{k^2,|\mathcal U|\})$
\end{lemma}

\begin{lemma}\label{lemma:build G}
	$\mathcal G_{S,k}$ can be computed in $\bigO(n+|\mathcal G_{S,k}| + k\cdot |\mathcal C|)$ expected time.
\end{lemma}

Putting our bounds together, we obtain:

\begin{corollary}\label{cor:final G size}
	
	Let $g = \min\{ n^2,nk^2,\sigma^{2k-1}k^2 \}$. Then, $\mathcal G_{S,k}$ takes $\bigO(g)$ space and can be built in $\bigO(n+g)$ expected time. 
	
\end{corollary}
\begin{proof}
	From Lemma \ref{lemma:size of G}, $|\mathcal G_{S,k}|\in \bigO(|\mathcal C|\cdot \min\{k^2,|\mathcal U|\})$. By Theorem \ref{th:main}, this space is $\bigO(\min\{n,\sigma^{2k-1} + 2k-2\}\cdot \min\{k^2,\min\{n, \sigma^k\}\})$. Since $k^2\in \bigO(\sigma^k)$ and $k\in \bigO(\sigma^{2k-1})$, this space simplifies to $\bigO(g)$.
	
	Finally, note that $k\cdot |\mathcal C| \in \bigO(k\cdot \min\{n,\sigma^{2k-1} + 2k-2\}) \subseteq \bigO(g)$, so the running time of Lemma \ref{lemma:build G} is $\bigO(n+g)$.
\end{proof}

\section{Faster algorithms}\label{sec: faster algorithms}

In this section we use properties of our reduction to provide faster algorithms for a range of problems: (i) checking that a given set $\Gamma\subseteq [1..n]$ is a $k$-attractor, (ii) checking that a given set is a \emph{minimal} $k$-attractor, (iii) finding a minimum $k$-attractor, (iv) finding a minimal $k$-attractor, and (v) approximate a minimum $k$-attractor. We note that problems (i)-(ii) admit naive cubic solutions, while problem (iii) is NP-hard for $k\geq 3$~\cite{kempa2017roots}.

\subsection{Checking the attractor property}

Given a string $S$, a set $\Gamma\subseteq [1..n]$, and an integer $k\geq 1$, is $\Gamma$ a $k$-attractor for $S$? we show that this question can be answered in $\bigO(n)$ time.

The main idea is to use Lemma \ref{theorem:k-attr characterization} and check, for every suffix tree edge $e$ at string depth at most $k$, if $\Gamma$ marks $e$.
Consider the suffix array $SA[1,n]$ of $S$ and the array $D[1,n]$ defined as follows: $D[i] = successor(\Gamma, SA[i])-SA[i]$, where $successor(X,x)$ returns the smallest element larger than or equal to $x$ in the set $X$ (i.e. $D[i]$ is the distance between $SA[i]$ and the element of $\Gamma$ following---and possibly including---$SA[i]$). 
$D$ can be built in linear time and space by creating a bit-vector $B[1,n]$ such that $B[i]=1$ iff $i\in \Gamma$ and pre-processing $B$ for constant-time successor queries~\cite{jacobson1988succinct,clark1998compact}.
We build a range-minimum data structure (RMQ) on $D$ ($\bigO(n)$ bits of space, constant query time~\cite{FH11}). Then for every suffix tree edge $e$ such that $\lambda(e)\leq k$, we check (in constant time) that $\lambda(e) > \min(D[l_e..r_e])$. The following lemma ensures that this is equivalent to checking whether $\Gamma$ marks $e$.

\begin{lemma}\label{lemma:check marking}
	$\lambda(e) > \min(D[l_e..r_e])$ if and only if $\Gamma$ marks $e$.
\end{lemma}
\begin{proof}
	$(\Rightarrow)$ Assume that $\lambda(e) > \min(D[l_e..r_e])$. By definition of $D$, this means that there exist an index $i' \in [l_e..r_e]$ and a  $j\in\Gamma$, with $j\geq i = SA[i']$, such that $j-i = D[i'] < \lambda(e)$. Equivalently, $i \leq j  < i + \lambda(e)$, i.e. $\Gamma$ marks $e$. 
	
	$(\Leftarrow)$ Assume that $\Gamma$ marks $e$. Then, by definition there exist an index $i' \in [l_e..r_e]$ and a  $j\in\Gamma$ such that $SA[i'] = i \leq j  < i + \lambda(e)$. Then, $j-i < \lambda(e)$. Since $D[i']$ is computed taking the $j\in\Gamma$, $j\geq SA[i']$, minimizing $j-SA[i']$, it must be the case that $D[i'] \leq j-i < \lambda(e)$. Since $i' \in [l_e..r_e]$, this implies that $\min(D[l_e..r_e]) < \lambda(e)$.
\end{proof}

Together, Lemmas \ref{theorem:k-attr characterization} and \ref{lemma:check marking} imply that, if $\lambda(e) > \min(D[l_e..r_e])$ for every edge at string depth at most $k$, then $\Gamma$ is a $k$-attractor for $S$. 
Since the suffix tree, as well as the other structures used by our algorithm, can be built in linear time and space on alphabet $[1..n]$~\cite{farach1997optimal} and checking each edge takes constant time, we obtain that the problem of checking whether a set $\Gamma\subseteq [1..n]$ is a valid $k$-attractor can be solved in optimal $\bigO(n)$ time and $\bigO(n)$ words of space. We now show how to improve upon this working space by using recent results in the field of compact data structures. In the following result, we assume that the input string is packed in $\bigO(n\log \sigma)$ bits (that is, $\bigO(n/\log_\sigma n)$ words).

We first need the following Lemma from \cite{belazzougui2016linear}:

\begin{lemma}{\cite[Thm. 3]{belazzougui2016linear}}\label{lemma:enumerate}
	In $\bigO(n)$ time and $\bigO(n\log \sigma)$ bits of space we can enumerate the following information for each suffix tree edge $e$:
	\begin{itemize}
		\item The suffix array range $\langle l_e, r_e\rangle$ of the string $s(e)$, and
		\item the length $\lambda(e)$ of $s(e)$.
	\end{itemize}
\end{lemma}

We can now prove our theorem. Note that the input set $\Gamma\subseteq [1..n]$ can be encoded in $n$ bits, so also the input fits in $\bigO(n\log\sigma)$ bits.

\begin{theorem}\label{th:checking}
	Given a string $S[1,n]$, a set $\Gamma\subseteq [1..n]$, and an
	integer $k\geq 1$, we can check whether $\Gamma$ is a $k$-attractor
	for $S$ in:
	\begin{itemize}
		\item Optimal $\bigO(n\log\sigma)$ bits of space and
		$O(n\log^{\epsilon} n)$ time, for any constant $\epsilon>0$, or
		\item $\bigO(n(\log\sigma+\log k))$ bits of space and  $\bigO(n)$ time. 
	\end{itemize}
\end{theorem}
\begin{proof}
	To achieve the first trade-off we will replace the $D$ array
	(occupying $\bigO(n \log n)$ bits) with a smaller data structure supporting random access to $D$. We start by replacing the standard
	suffix array with a compressed suffix array (CSA)~\cite{FerraginaM05,GrossiV05}. Given a text stored in
	$\bigO(n\log\sigma)$ bits, the CSA can be built in deterministic
	$\bigO(n)$ time and optimal $\bigO(n\log\sigma)$ bits of
	space~\cite{MunroNN17}, and supports access queries to the suffix array $SA$ in
	$\bigO(\log^{\epsilon}n)$ time~\cite{GrossiV05}, for any constant $\epsilon>0$ chosen at construction time. Given that
	$D[i]=successor(\Gamma,SA[i])-SA[i]$ and we can compute the
	successor function in constant time using a $\bigO(n)$-bit data
	structure (array $B$), $D[i]$  can be computed in $\bigO(\log^{\epsilon} n)$
	time. Using access to $D$, the RMQ data structure (occupying $\bigO(n)$
	bits) can be built in $\bigO(n \log^{\epsilon} n)$ time and
	$\bigO(n)$ bits of space~\cite[Thm. 5.8]{FH11}. 
	At this point, observe
	that the order in which we visit suffix tree edges does not affect
	the correctness of our algorithm. By using Lemma \ref{lemma:enumerate} we can enumerate $\lambda(e)$ and $\langle l_e, r_e\rangle$ for every suffix tree edge $e$ in linear time and compact space, and check $\lambda(e) > \min(D[l_e..r_e])$ whenever $\lambda(e) \leq k$ (Lemma \ref{lemma:check marking}).

	To achieve the second trade-off we observe that in our algorithm
	we only explore the suffix tree up to depth $k$ (i.e. we only perform the check of Lemma \ref{lemma:check marking} when $\lambda(e)\leq k$), hence any $D[i]>k$
	can be replaced with $D[i]=k+1$ without affecting the correctness of the verification
	procedure. In this way,  array $D$ can be stored in just $O(n \log
	k)$ bits. To compute the $D$ array in $\bigO(n)$ time and compact space we observe that it suffices to access
	all pairs $\langle i,SA[i]\rangle$ in \emph{any} order (not necessarily
	$\langle1,SA[1]\rangle, \langle 2, SA[2]\rangle, \dots$). 
	From~\cite[Thm. 10]{belazzougui2016linear}, in $O(n)$ time and $O(n\log\sigma)$ bits of space we can
	build a compressed suffix array supporting constant-time LF function computation. By repeatedly applying LF from the first suffix array position, we 
	enumerate entries of the inverse suffix array $ISA$ in right-to-left order in $O(n)$ time~\cite[Lem. 1]{belazzougui2016linear}.
	This yields the sequence of pairs $\langle ISA[i],i\rangle = \langle j,SA[j]\rangle$, for $i=n, \dots, 1$ and $j=ISA[i]$, which can be used to compute $D$ in linear time and compact space. As in the first trade-off, we use Lemma \ref{lemma:enumerate} to enumerate $\lambda(e)$ and $\langle l_e, r_e\rangle$ for every suffix tree edge $e$, and check $\lambda(e) > \min(D[l_e..r_e])$  whenever $\lambda(e) \leq k$ (Lemma \ref{lemma:check marking}).
\end{proof}

Note that with the second trade-off of Theorem \ref{th:checking} we achieve $\bigO(n)$ time and optimal $\bigO(n\log\sigma)$-bits of space when $k\leq \sigma^{\bigO(1)}$ (in particular, this is always the case when $k$ is constant).
Note also that, since we now assume that the input string is packed in $\bigO(n/\log_\sigma n)$ words, the running time is not optimal (being $\Omega(n/\log_\sigma n)$ a lower-bound in this model).

As a by-product of Theorem \ref{th:checking}, we obtain a data structure able to report, given a substring $s$, all occurrences $s'$ of $s$ straddling a position in $\Gamma$.

\begin{theorem}\label{th:reporting}
	Let $S[1,n]$ be a string and $\Gamma$ be an attractor for $S$.
	In $\bigO(n)$ space we can build a structure of size $\bigO(n)$ words supporting the following query: given a range $\langle i,j \rangle$, report all (or at most any number of) positions $i'$ such that $S[i'..i'+(j-i)] = S[i..j]$ and $j' \in [i'..i'+(j-i)]$ for some $j'\in \Gamma$. Every such $i'$ is reported in constant time. 
\end{theorem}

\subsection{Checking minimality}

Given a string $S$, a set $\Gamma\subseteq [1..n]$, and an integer $k\geq 1$, is $\Gamma$ a \emph{minimal} $k$-attractor for $S$? The main result of this section is that this question can be answered in almost-optimal $\bigO(n\log n)$ time.

We first show that minimal $k$-attractors admit a convenient characterization based on the concept of suffix tree $k$-marking.

\begin{definition}[$k$-necessary position]\label{def:k-necessary}
	$j\in \Gamma$ is \emph{$k$-necessary} with respect to a set $\Gamma'$, with $\Gamma \subseteq \Gamma' \subseteq [1..n]$, if and only if there is at least one suffix tree edge $e$ such that:
	\begin{enumerate}
		\item $\lambda(e)\leq k$,
		\item $j$ marks $e$, and
		\item If $j'\in\Gamma'$ marks $e$, then $j'=j$
	\end{enumerate}
\end{definition}

\begin{definition}[$k$-necessary set]\label{def:k-necessary-set}
	$\Gamma$ is \emph{$k$-necessary} if and only if all its elements are $k$-necessary with respect to $\Gamma$.
\end{definition}

A remark: we give Definition \ref{def:k-necessary} with respect to a general superset $\Gamma'$ of $\Gamma$, but for now (Definition \ref{def:k-necessary-set}) we limit our attention to the case $\Gamma' = \Gamma$. Later (Theorem \ref{theorem:find minimal}) we will need the more general definition. For simplicity, in the proofs of the following two theorems we just say \emph{$k$-necessary} (referring to some $j\in\Gamma$) instead of \emph{$k$-necessary with respect to $\Gamma$}.

\begin{lemma}\label{theorem:minimal characterization}
	$\Gamma$ is a minimal $k$-attractor if and only if: 
	\begin{enumerate}
		\item It is a $k$-attractor, and
		\item it is $k$-necessary.
	\end{enumerate}
\end{lemma}

A naive solution for the minimality-checking problem is to test the $k$-attractor property on $\Gamma - \{i\}$ for every $i\in \Gamma$ using Theorem \ref{th:checking}. This solution, however, runs in quadratic time. 	
Our efficient strategy consists in checking, for every suffix tree edge $e$, if there is only one $j\in\Gamma$ marking it. In this case, we flag $j$ as necessary. If, in the end, all attractor positions are flagged as necessary, then the attractor is minimal by Lemma \ref{theorem:minimal characterization}. 

Unfortunately, the following simple strategy does not work: for every suffix tree edge $e$, report, with the structure of Theorem \ref{th:reporting}, two distinct occurrences of the string $s(e)$ that straddle an attractor position. The reason why this does not work is that, even we find two such occurrences, the attractor position that they straddle could be the same (this could happen e.g. if $s$ is periodic). 
Our solution is, therefore, a bit more involved.

\begin{theorem}\label{theorem:minimality check}
	Given a string $S[1,n]$, a set $\Gamma\subseteq [1..n]$, and an integer $k\geq 1$, we can check whether $\Gamma$ is a minimal $k$-attractor for $S$ in $\bigO(n\log n)$ time and $\bigO(n\log |\Gamma|)$ space. 
\end{theorem}

\subsection{Computing a minimum $k$-attractor}

Computing a minimum $k$-attractor is NP-hard for $k\geq 3$ and general $\sigma$. In this section we show that the problem is actually polynomial-time-solvable for small $k$ and $\sigma$. Our algorithm takes advantage of both our reduction to set-cover and the optimal verification algorithm of Theorem \ref{th:checking}.

First, we give an upper-bound to the cardinality of the set of all minimal $k$-attractors. This will speed up our procedure for finding a  minimum $k$-attractor (which must be, in particular, minimal). 
By Lemma \ref{lemma:size}, there are no more than
$exp(\bigO(\sigma^{2k}))$ $k$-attractors for $S$. With the following lemma, we give a better upper-bound to the number of \emph{minimal} $k$-attractors.

\begin{lemma}\label{lemma: number minimal k-attractors}
	There cannot be more than
	$exp\left(\bigO(\sigma^k\log\sigma^k)\right)$ 
	minimal $k$-attractors.
\end{lemma}

Using the above lemma, we now provide  a strategy to find a minimum $k$-attractor.

\begin{theorem}\label{th:minimum}
	A minimum $k$-attractor
	can be found in $\bigO(n) + exp\big(\bigO(\sigma^k\log\sigma^k)\big)$ time. 
\end{theorem}

\subsection{Computing a minimal $k$-attractor}\label{sec:minimal}

In this section we provide an algorithm to find a minimal $k$-attractor, and then show that such a solution is a $k$-approximation to the optimum.

\begin{theorem}\label{theorem:find minimal}
	A minimal $k$-attractor
	can be found in $\bigO(\min\{n^2,nk^2, n+\sigma^{2k-1} k^2\})$ expected time. 
\end{theorem}

We now give the approximation ratio of minimal attractors, therefore showing that the strategy of Theorem \ref{theorem:find minimal} yields an approximation algorithm.

\begin{theorem}\label{theorem:minimal approx rate}
	Any minimal $k$-attractor is a $k$-approximation to the minimum $k$-attractor.
\end{theorem}

\subsection{Better approximations to the minimum $k$-attractor}\label{sec:approx}

From~\cite{kempa2017roots}, we can compute poly-logarithmic approximations to the smallest attractor in linear time using reductions from dictionary compression techniques. This strategy, however, works only for $n$-attractors. 

In~\cite[Thm. 10]{kempa2017roots}, the authors show that a simple reduction to $k$-set cover allows one to compute in polynomial time a $\mathcal H(k(k+1)/2)$-approximation to the minimum $k$-attractor, where $\mathcal H(p) = \sum_{i=1}^p \frac{1}{i} \leq \ln (p+1) +1$ is the $p$-th harmonic number. This approximation ratio is at most $1+2\ln\left( \frac{k+1}{\sqrt 2} \right)$ for $k>1$ (and case $k=1$ is trivial to solve optimally in linear time). 
The key observation of~\cite[Thm. 10]{kempa2017roots} is that we can view each text position $i$ as a set $s_i$ containing all distinct $k'$-mers, with $k'\leq k$, overlapping the position. Then, solving $k$-attractor is equivalent to covering the universe set of all distinct substrings of length at most $k$ using the smallest possible number of sets $s_i$. This is, precisely, a $(k(k+1)/2)$-set cover instance (since $|s_i| \leq k(k+1)/2$ for all $i$), which can be approximated within a factor of  $\mathcal H(k(k+1)/2)$ using the greedy algorithm that at each step chooses the $s_i$ that covers the largest number of uncovered universe elements. A naive implementation of this procedure, however, runs in cubic time. 
We now show how to efficiently implement this greedy algorithm over the reduction of Theorem \ref{th:main}.

We first give a lemma needed to achieve our result.

\begin{lemma}\label{lemma:queue}
	Let $U$ be some universe of elements and $p:U\rightarrow [1..M]$ be a function assigning a priority to each element of $U$.
	In $\bigO(M+q)$ time we can build a priority queue $\mathcal Q$ initialized with $q$ elements of $U$ such that, later, we can perform any sequence of $m$ of the following operations in expected  $\bigO(M+m)$ time:
	
	\begin{itemize}
		\item $\mathcal Q.pop()$: return the $x\in \mathcal Q$ with largest $p(x)$ and remove $x$ from $\mathcal Q$.
		\item $\mathcal Q.dec(x)$: update $p(x) \leftarrow p(x)-1$, for a $x\in \mathcal Q$.
	\end{itemize}	
	
	At any point in time, the size of $\mathcal Q$ is of $\bigO(M+q)$ words.
\end{lemma}

\begin{theorem}\label{theorem:approx}
	A $\mathcal H(k(k+1)/2)$-approximation of a minimum $k$-attractor can be computed in $\bigO(\min\{n^2, nk^2, n+\sigma^{2k-1} k^2\})$ expected time, where  $\mathcal H(p) = \sum_{i=1}^{p} \frac{1}{i}$ is the $p$-th harmonic number.
\end{theorem}

Note that the approximation ratio of Theorem \ref{theorem:approx} is $\mathcal H(k(k+1)/2) \leq 1+2\ln\left( \frac{k+1}{\sqrt 2} \right) \in \bigO(\log k)$.

\section{$k$-sharp attractors}
\label{section:sharp}

In this section we consider a natural variant of string attractors we
call \emph{$k$-sharp attractors}, and we prove some results concerning
their computational complexity.

Formally, we define a $k$-\emph{sharp} attractor of a string $S\in
\Sigma^n$ to be a set of positions $\Gamma\subseteq\left[1..n\right]$ such that every substring $S[i..j]$ with $j-i+1 = k$ has
an occurrence $S[i'..j'] = S[i..j]$ with $j'' \in [i'..j']$ for some
$j''\in\Gamma$. In other words, a $k$-sharp-attractor is a subset that
covers all substrings of length exactly $k$.

By \textsc{Minimum-$k$-Sharp-Attractor} we denote the optimization
problem of finding the smallest $k$-sharp attractor of a given input
string.  By $k\textsc{-Sharp-Attractor} = \{\langle T,p\rangle :
\text{String }T\text{ has a }k\text{-sharp-attractor of size}\leq p\}$
we denote the corresponding decision problem.  The NP-completeness of
\textsc{$k$-Sharp-Attractor} for constant $k$ is obtained by a
reduction from \textsc{$k$-SetCover} problem that is
NP-complete~\cite{DuhF97} for any constant $k \geq 3$: given integer
$p$ and a collection $C=\{C_1, C_2, \ldots, C_m\}$ of $m$ subsets of
the universe set $\mathcal{U}=\{1, 2, \ldots, n\}$ such that
$\bigcup_{i=1}^{m}C_i=\mathcal{U}$, and for any $i\in\{1,\ldots,m\}$,
$|C_i| \leq k$, return YES iff there exists a subcollection
$C'\subseteq C$ such that $\bigcup C' = \mathcal{U}$ and $|C'|\leq p$.

\begin{theorem}
	\label{th:sharp-attractor-npc}
	For any constant $k \geq 3$, \textsc{$k$-Sharp-Attractor} is
	NP-complete.
\end{theorem}
\begin{proof}[Proof idea]
	We obtain our reduction as follows (see the Appendix for full details). 
	For any constant $k \geq 3$,
	given an instance $\langle \mathcal{U}, C\rangle$ of
	\textsc{$k$-SetCover} we build a string $S_C$ of length
	$\bigO(mk^2)$, where $m$  denotes the number of sets in the input collection, with the following property: $\langle \mathcal{U},C
	\rangle$ has a cover of size $\leq p$ if and only if $S_C$ has a
	$k$-sharp-attractor of size $\leq 2t+m+p+2$, where
	$t=\sum_{i=1}^{m}n_i$. The proof
	follows~\cite[Thm. 8]{kempa2017roots}, but the main gadgets are
	slightly different.
	
	
	The main idea is to let each element $u\in U$ correspond to a
	substring, which is repeated once for every set $C_i$ containing
	$u$. We construct a substring $S_i$ for each set $C_i$ such that
	$S_i$ can always be covered by $2|C_i| + 2$ attractor positions,
	corresponding to choosing the set, but $S_i$ can be covered by
	$2|C_i| +1$ attractor positions when all elements of the set already
	belong to chosen sets. And, $S_i$ indeed requires $2|C_i| + 1$
	attractor positions. Finally, the substrings $S_i$ are padded and
	concatenated to form one long string $S_C$. 
\end{proof}

Denote the size of the minimum $k$-sharp-attractor of string $S$ by
$\gamma_{k}^{*}(S)$. Observe that the above theorem is proved for
constant values of $k$. This is because, unlike for general
$k$-attractors, $\gamma_{k}^{*}(S)$ is not monotone with respect to
$k$. This becomes apprarent when we observe that
$\gamma_{|S|}^{*}(S)=1$ for all $S$, hence e.g. for $X={\tt abb}$ we
have $\gamma_{2}^{*}(X) \not\leq \gamma_{3}^{*}(X)$. Note also that
our reduction requires large alphabet.

Interestingly, however, for $k=2$ the $k$-sharp-attractor admits a
polynomial-time algorithm. Note that such a result is not known for $k$-attractors (the case $k=2$ being the only one still open).

\begin{theorem}
	\label{thm:2-sharp-attractor}
	Minimum $2$-sharp-attractor is in P.
\end{theorem}

\begin{proof}
	It is easy to show that $2$-sharp-attractor is in $P$ by a reduction
	to edge cover. Given a string $S$, let $V\subseteq \Sigma^2$ be the
	set of strings of length $2$ that occur at least once in $S$. For
	every substring of length $3$ of the form $xyz$, add the edge
	$(xy,yz)$ to the edge-set $E$, and add self-loops for the first and
	last pair.
	
	A position $\gamma\in \Gamma$ thus corresponds to an edge,
	$e_\gamma$, and it is easy to see that $\Gamma$ is a
	$2$-sharp-attractor if and only if $\{e_\gamma | \gamma \in \Gamma
	\}$ is an edge cover.
	
	The number of vertices and edges in this graph are both $\leq n$, so
	a minimum edge cover can be found in $O(n\sqrt{n})$
	time~\cite{Micali1980}.
\end{proof}

\bibliography{attractors-algo}

\newpage

\appendix

\section{Appendix}

\begin{proof}[Proof of Lemma \ref{lemma:size}]
	By the way we defined $\equiv _k$, the set $[1..n]/\equiv _k$ has one element per distinct substring of length $(2k-1)$ in $S'$, that is, per distinct path from the suffix tree root to each of the nodes in $\mathcal L(ST^{2k-1}(S'))$.
	Clearly, $|\mathcal L(ST^{2k-1}(S'))|\leq n$.
	On the other hand, there are at most $\sigma^{2k-1}$ distinct substrings of length $2k-1$ on $\Sigma$. There are other $2k-2$ additional substrings to consider on the borders of $S'$ (to include the runs of symbol $\#$). It follows that the cardinality of $\mathcal L(ST^{2k-1}(S'))$ is upper-bounded also by $\sigma^{2k-1} + 2k-2$. 
\end{proof}

\begin{proof}[Proof of Lemma \ref{lemma:equiv_classes}]
	Suppose, by contradiction, that $|\Gamma \cap [i]_{\equiv_k}| > 1$ for some $i$. Then, let $j,j'\in \Gamma \cap [i]_{\equiv_k}$, with $j\neq j'$. By definition of $\equiv_k$, $S'[j-k+1..j+k-1] = S'[j'-k+1..j'+k-1]$. This means that if a substring of $S$ of length at most $k$ has an occurrence straddling position $j$ in $\Gamma$ then it has also one occurrence straddling position $j' \in (\Gamma - \{j\})$. On the other hand, any other substring occurrence straddling any position $j''\neq j,j'$ is also captured by $\Gamma - \{j\}$ since $j''$ belongs to this set.	
	This implies that $\Gamma - \{j\}$ is a $k$-attractor, which contradicts the minimality of $\Gamma$.
\end{proof}

\begin{proof}[Proof of Lemma \ref{lemma:size of G}]
	Clearly, $|E| \in \bigO(|\mathcal C|\cdot |\mathcal U|)$.
	Note that each position $j\in \mathcal C$ crosses at most $\bigO(k^2)$ distinct substrings of length at most $k$ in $S$, therefore this is also an upper-bound to the number of suffix tree edges it can mark. It follows that there are at most $\bigO(k^2)$ edges in $\mathcal G_{S,k}$ sharing a fixed $j\in \mathcal C$, which implies $|E| \in \bigO(|\mathcal C|\cdot k^2)$.
	
\end{proof}

\begin{proof}[Proof od Lemma \ref{lemma:build G}]
	Clearly, $\mathcal U$ can be computed in linear time using the suffix tree of $S$, as this set contains suffix tree edges at string depth at most $k$.
	
	The set $\mathcal C$ can be computed considering all suffix tree nodes (explicit or implicit) at string depth $2k-1$, extracting the leftmost occurrence in their induced sub-tree (i.e. an occurrence of the string of length $2k-1$ read from the root to the node), and adding $k-1$. This task takes linear time once the suffix tree of $S$ is built.

	Let $st\_edge(S[i..i+k-1])$ return the suffix tree edge $e$ reached following the path labeled $S[i..i+k-1]$ from the suffix tree root. This edge can be found in constant time using the optimal structure for weighted ancestors queries on the suffix tree described in~\cite{gawrychowski2014weighted} (which can be build beforehand in $\bigO(n)$ time and space).
	Let moreover $\pi(e)$ be the parent edge of $e$, i.e. $\pi(e) = \langle u,v\rangle$ and $e = \langle v, u'\rangle$ for some suffix tree nodes $u,v,u'$.
	We implement $E$ using hashing, so insert and membership operations take expected constant time.
	
	To build $E$ we proceed as follows. 
	For every $j\in \mathcal C$, and for $i=j, j-1, \dots, j-k+1$:
	\begin{enumerate}
		\item Find edge $e = st\_edge(S[i..i+k-1])$.
		\item If $\lambda(e) > j-i$ and $\langle j, e \rangle \notin E$, then insert $\langle j, e\rangle$ in $E$. Otherwise, proceed at step 1 with the next value of $i$.
		\item $e \leftarrow \pi(e)$. Repeat from step 2.
	\end{enumerate}

	\textbf{Correctness} It is easy to see that we only insert in $E$ edges $\langle j,e\rangle$ such that $j$ marks $e$ (because we check that $\lambda(e) > j-i$), so the algorithm is correct.
	
	\textbf{Completeness} We now show that if $j$ marks $e$, then we insert $\langle j,e\rangle$ in $E$. Assume that  $j$ marks $e$. Consider the (unique) occurrence $S[i..i+\lambda(e)-1] = s(e)$ overlapping $j$ in its leftmost position (i.e. $i \leq j < i+\lambda(e)$ and $j-i$ is minimized; note that there could be multiple occurrences of $s(e)$ overlapping $j$). 
	Consider the moment when we compute $e'= st\_edge(S[i..i+k-1])$ at step 1. We now show that, for each edge $e''$ on the path from $e'$ to $e$, it must hold  $\langle j,e''\rangle\notin E$. This will prove the property, since then we insert all these edges (including $\langle j,e\rangle$) in $E$.
	Assume, by contradiction, that $\langle j,e''\rangle\in E$ for some $e''$ on the path from $e'$ to $e$. Then this means that in the past at step 1. we have already considered some $S[i'..i'+k-1]$, with $i < i' \leq j < i'+k$, prefixed by $s(e'')$ (note that it must be the case that $i'>i$ since we consider values of $i$ in decreasing order). But then, since $s(e)$ prefixes $s(e'')$, it also prefixes $S[i'..i'+k-1]$, i.e. $S[i'..i'+\lambda(e)-1] = s(e)$ and $i'>i$. This is in contradiction with the way we defined $i$ (i.e. $j-i$ is minimized).
	
	\textbf{Complexity} 
	Overall, in step 1. we call $\bigO(k\cdot |\mathcal C|)$ times function $st\_edge$ (constant time per call). Then, in steps 2. and 3. we only spend (constant) time whenever we insert a new edge in $E$ (since we check $\langle j, e \rangle \notin E$ before inserting). Overall, our algorithm runs in $\bigO(n+|\mathcal G_{S,k}| + k\cdot |\mathcal C|)$ time.
\end{proof}

\begin{proof}[Proof of Lemma \ref{lemma:enumerate}]
	In \cite[Thm. 3]{belazzougui2016linear} (see also \cite{Belazzougui14}) the authors show how to enumerate the following information for each right-maximal substring $W$ of $S$ in $\bigO(n)$ time and $\bigO(n\log \sigma)$ bits of space: $|W|$ and the suffix array range $range(Wb)$ of the string $Wb$, for all $b\in\Sigma$ such that $Wb$ is a substring of $S$. Since $W$ is right-maximal, those $Wb$ are  equal to our strings $s(e)$ (for every edge $e$). It follows that our problem is solved by outputting all $range(Wb)$ and $|Wb|$ returned by the algorithm in~\cite[Thm. 3]{belazzougui2016linear}.
\end{proof}

\begin{proof}[Proof of Theorem \ref{th:reporting}]
	
	The idea is to use the variant of Theorem \ref{th:checking} based on the (uncompressed) suffix tree together with the optimal structure for weighted ancestors queries on the suffix tree described in~\cite{gawrychowski2014weighted}. Weighted ancestors on the suffix tree can be used to find the suffix tree edge $e$ where string $s = S[i..j]$ ends (starting from the root) in constant time given as input the range $\langle i,j\rangle$~\cite{gawrychowski2014weighted}. Reporting an arbitrary number of occurrences of $s=S[i..j]$ straddling a position in $\Gamma$ corresponds to solving a three-sided orthogonal range reporting query:
	we find the minimum $D[m]$ in $D[l_e..r_e]$ and check if $|s| = j-i+1 > D[m]$. If the answer is yes, we output $SA[m]$ and recurse on $D[l_e..m-1]$ and $D[m+1..r_e]$. If the answer is no, we stop. Note that we can report an arbitrary number of such occurrences with constant delay. Since the structure~\cite{gawrychowski2014weighted} can be built in linear time and space, we obtain our claim.
\end{proof}

\begin{proof}[Proof of Lemma \ref{theorem:minimal characterization}]
	$(\Rightarrow)$ Let $\Gamma$ be a minimal $k$-attractor.
	Let $j\in \Gamma$. Since $\Gamma$ is minimal, $\Gamma - \{j\}$ is not a $k$-attractor. From Theorem \ref{theorem:k-attr characterization}, this implies that $\Gamma - \{j\}$ is not a $k$-marking, i.e. there exists a suffix tree edge $e$, with $\lambda(e)\leq k$, that is not marked by $\Gamma - \{j\}$. On the other hand, the fact that $\Gamma$ is a $k$-attractor implies (Theorem \ref{theorem:k-attr characterization}) that $\Gamma$ is a $k$-marking, i.e. it also marks edge $e$. This, in particular, implies that $j$ marks $e$.
	Now, let $j'\in \Gamma$ be a position that marks $e$. Assume, by contradiction, that $j'\neq j$. Then, $j' \in \Gamma - \{j\}$, which implies that $\Gamma - \{j\}$ marks $e$. This is a contradiction, therefore it must be the case that $j'=j$, i.e. $j$ is $k$-necessary. Since the argument works for any $j\in \Gamma$, we obtain that all $j\in \Gamma$ are $k$-necessary.
	
	$(\Leftarrow)$ Assume that $\Gamma$ is a $k$-attractor and all $j\in \Gamma$ are $k$-necessary. Then, choose an arbitrary $j\in \Gamma$. 
	By Definition \ref{def:k-necessary}, there exists an edge $e$ that is only marked by $j$, i.e. for every $j'\in \Gamma - \{j\}$, $j'$ does not mark $e$.
	This implies (Theorem \ref{theorem:k-attr characterization}) that $\Gamma - \{j\}$ is not a $k$-attractor. Since the argument works for any $j\in \Gamma$, we obtain that $\Gamma$ is a minimal $k$-attractor.	
\end{proof}

\begin{proof}[Proof of Theorem \ref{theorem:minimality check}]
	We associate to each element in $\Gamma$ a distinct color from the set $\Sigma_\Gamma = \{c_i\ :\ i\in\Gamma\}$, and we build a two-dimensional $\Sigma_\Gamma$-colored grid $L \subseteq [1..n]^2\times \Sigma_\Gamma$ (i.e. each point $\langle i,j\rangle$ in $L$ is associated with a color from $\Sigma_\Gamma$) defined as $L = \{ \langle i,D[i], c_{SA[i]+D[i]} \rangle,\ i=1,...,n  \}$, that is, at coordinates $\langle i,D[i]\rangle$ we insert a point ``colored'' with the color associated to the attractor position immediately following---and possibly including---$SA[i]$. Then, for every suffix tree edge $e$ we check that the range $L[l_e..r_e, 0..\lambda(e)-1]$ contains at least two distinct colors. If there are at least two distinct colors $c_k\neq c_{k'}$ in the range, then we do not mark $k$ and $k'$ as necessary (note that they could be marked later by some other edge, though). However, if there is only one distinct color $c_k$ in the range then this is not enough to mark $k$ as necessary. The reason why is that in array $D$ we are tracking only the attractor position $i'$ \emph{immediately following} each text position $i$; it could well be that the attractor position $i'' > i'$ immediately following $i'$ marks $e$, but we miss it because we track only $i'$. This problem can be easily solved inserting in $L$ also a point corresponding to the \emph{second} nearest attractor position following every text position (so the number of points only doubles). It is easy to see that this is sufficient to solve our problem, since we only aim at enumerating at most two distinct colors in a range.
	
	At this point, we have reduced the problem to the so-called \emph{three-sided colored orthogonal range reporting} problem in two dimensions: report the distinct colors inside an three-sided orthogonal range in a grid. 
	For this problem, the fastest known data structure takes $\bigO(n\log n)$ space and answers queries in $\bigO(\log^2n + i)$ time, where $n$ is the number of points in the grid and $i$ is the number of returned points~\cite{gupta1995further}. This would result in an overall running time of $\bigO(n\log^2n)$ for our algorithm.
	We note that our problem is, however, simpler than the general one. In our case, it is enough to list \emph{two} distinct colors (if any); we are not interested in counting the total number of such colors or reporting an arbitrary number of them. 
	
	Our solution relies on wavelet trees~\cite{grossi2003high}. First, we pre-process the set of $v \leq 2n$ points so that they fit in a grid $[1..v]\times[1..v]$ such that every row and every column contain exactly one point. Mapping the original query on this grid can be easily done in constant time using well-established rank reduction techniques that we do not discuss here (see, e.g.~\cite{navarro2016compact}). We can view this grid as an integer vector $V \in [1..v]^v$, where each $V[i]$ is associated with a color $V[i].c\in\Sigma_{\Gamma}$.  We build a wavelet tree $WT(V)$ on $V$. Let $u_{s}$ denote the internal node of $WT(V)$ reached following the path $s\in\{0,1\}^*$. With $V_s$ we denote the subsequence of $V$ associated with $u_s$, i.e. the subsequence of $V$ such that the binary representation of each $V_s[i]$ is prefixed by $s$.
	For each internal node $u_s$ we store (i) the sequence of colors $C_{s} = V_s[1].c, \dots, V_s[|V_s|].c$, and (ii) a bitvector $B_s[1..|V_s|]$ such that $B_s[1]=0$ and $B_s[i] \neq B_s[i-1]$ iff $C_s [i] \neq C_s[i-1]$. We pre-process  each $B_s$ for constant-time rank and select queries~\cite{jacobson1988succinct,clark1998compact}. Overall, our data structure takes $\bigO(n\log \Gamma)$ words of space (that is, $\bigO(n\log \Gamma\log n)$ bits: at each of the $\log n$ levels of the wavelet tree we store $v\leq 2n$ colors). 
	
	To report two distinct colors in the range $[l,r]\times[l',r']$, we find in $\bigO(\log v)$ time the $\bigO(\log v)$ nodes of $WT(V)$ covering $[l',r']$ as usually done when solving orthogonal range queries on wavelet trees (see~\cite{navarro2014wavelet} for full details). For each such node $u_s$, let the range $V_s[l_s, r_s]$ contain the elements in common between $V_s$ and $V[l..r]$ (i.e. the range obtained mapping $[l..r]$ on $V_s$). We check whether in $B_s[l_s..r_s]$ there are two distinct bits at adjacent positions. 
	If this is the case, we locate their positions $B_s[i_0]=0$ and $B_s[i_1]=1$, with $l_s \leq i_0 = i_1-1 \leq r_s$ (the case $i_1=i_0-1$ is symmetric), and return the colors $C_s[i_0]$ and $C_s[i_1]$. By construction of $B_s$, $C_s[i_0]\neq C_s[i_1]$ and therefore we are done. Note that $i_0$ and $i_1$ can be easily found in constant time using rank/select queries on $B_s$.
	
	If, on the other hand, all sequences $B_s[l_s..r_s]$ are unary, then we just need to retrieve the $\bigO(\log v)$ colors $C_s[l_s]$, for all the $u_s$ covering the interval  $[l',r']$, and check if any two of them are distinct (e.g. radix-sort them in linear time). Also this step runs in $\bigO(\log v)$ time. 
	
	Overall, our solution uses $\bigO(n\log|\Gamma|)$ space and runs in $\bigO(n\log v) = \bigO(n\log n)$ time. 
	
\end{proof}

\begin{proof}[Proof of Lemma \ref{lemma: number minimal k-attractors}]
	Let $minimal(\sigma, k)$ denote the maximum number of minimal $k$-attractors on the alphabet $[1..\sigma]$ (independently of the string length $n$).
	Let $\Gamma$ be a minimal $k$-attractor. 
	By Lemma \ref{theorem:minimal characterization}, for every $j\in\Gamma$ there is at least one edge $e\in\mathcal U$ marked by $j$ only.
	Let $edge:\Gamma \rightarrow \mathcal U$ be the function defined as follows: $edge(j) = e$ such that (i) $e$ is marked  by $j$ only, and (ii) among all edges marked  by $j$ only, $e$ is the one with the lexicographically smallest $s(e)$ (where, if $s(e')$ prefixes $s(e'')$, then we consider $s(e')$ smaller than $s(e'')$ in lexicographic order). 
	Let $\mathcal U' = edge(\Gamma)$ be the image of $\Gamma$ through $edge$. By its definition, $edge$ is a bijection between $\Gamma$ and $\mathcal U'$. This implies that $|\Gamma| = |\mathcal U'| \leq |\mathcal U| \leq \sigma^k$: a minimal $k$-attractor is a set of cardinality at most $\sigma^k$ chosen from a universe $\mathcal C$ of size at most $|\mathcal C| \leq \sigma^{2k-1}+2k-2 \leq \sigma^{2k}$, therefore:
	$minimal(\sigma, k) \leq \sum_{i=1}^{\sigma^k} {{\sigma^{2k}}\choose{i}}
	$.
	We now give an upper-bound to the function $f(N,t) = \sum_{i=1}^{t} {{N}\choose{i}}$, where we assume $t\geq 2$ for simplicity (the hypothesis holds in our case since $t=\sigma^k$). Then, we will plug our bound in the above inequality. Since ${{N}\choose{i}} < \frac{N^i}{i!}$, we have that 
	$$
	\begin{array}{lcl}
	f(N,t) &<& \sum_{i=1}^{t} \frac{N^i}{i!}\\ 
	&=& \sum_{i=1}^{t} \left(\frac{N}{t}\right)^i\cdot \frac{t^i}{i!} \\
	&\leq& \sum_{i=1}^{t} \left(\frac{N}{t}\right)^i\cdot \sum_{i=1}^{t} \frac{t^i}{i!}\\
	&\in & \bigO\left(\left(\frac{N\cdot e}{t}\right)^t\right)
	\end{array}
	$$
	We obtain our claim: $
	minimal(\sigma, k) \leq f(\sigma^{2k},\sigma^k) \in \bigO(e^{\sigma^k}\cdot \sigma^{k\sigma^k}) \leq exp\left(\bigO(\sigma^k\log\sigma^k)\right)$.
\end{proof}

\begin{proof}[Proof of Theorem \ref{th:minimum}]
	Let $c(i) = S'[ i-k+1..i+k-1 ]$, where $S'[i] = \# \notin \Sigma$ if $i<1$ or $i>n$ and $S'[i] = S[i]$ otherwise, be the context string associated to position $i$. Consider the string 
	$$
	C = c(i_1)\$c(i_2)\$\dots\$c(i_t)
	$$
	where $\{i_1, i_2, \dots, i_t\} = \mathcal C$ and $\# \neq \$\notin \Sigma$. By our choice of $\mathcal C$, the length of this string is $|C| = (|\mathcal C|\cdot 2k) - 1 \leq (\sigma^{2k-1} + 2k)\cdot 2k\in\bigO(\sigma^{2k}\cdot k) \leq exp(\bigO(\log\sigma^k))$. We can build $C$ in $\bigO(n+|C|)$ time using the suffix tree of $S$ (i.e. extracting all paths from the root to nodes at string depth at most $2k-1$).
	
	Let now $\Gamma'' \subseteq \{k\cdot (2j+1)\ :\ j=0,\dots,t-1 \}$. It is easy to see that $\Gamma' = \{i\ :\ C[i]=\$\ or\ C[i]=\#\}\ \cup \Gamma''$ is a $k$-attractor for $C$ if and only if the set $\Gamma = \{ i_{(x-k)/(2k)}\ :\ x\in \Gamma''\}$ is a $k$-attractor for $S$. Suppose that $\Gamma'$ is a $k$-attractor for $C$, and consider a substring $s$ of $S$ of length at most $k$. By construction of $C$, $s$ is also a substring of $C$; in particular, there is an occurrence $C[i..i+|s|-1] = s$ straddling a position $k\cdot(2j+1)\in\Gamma''$, for some $j$. Then, $i_j\in \Gamma$ and, by the way we defined $C$, there is an occurrence of $s$ straddling position $i_j$ in S. Conversely, suppose that $\Gamma$ is a $k$-attractor for $S$, and let $s$ be a substring of $C$ of length at most $k$. If $s$ contains either $\$$ or $\#$, then it must straddle one of the positions in $\{i\ :\ C[i]=\$\ or\ C[i]=\#\}\ \subseteq \Gamma'$. Otherwise, it appears inside one of the substrings $c(i_k)$, for some $k\in[1..t]$. But then, this means that $s$ appears in $S$ and, in particular, that it has some occurrence $s'=s$ straddling a position $j\in \Gamma$. By the way we constructed $C$, $s$ has an occurrence in $C$ straddling position $k\cdot(2j+1) \in \Gamma'' \subseteq \Gamma'$.

	At this point, we check whether $\Gamma'$ is a $k$-attractor for $C$ for all possible $\Gamma''$, and return the smallest such set. 
	Instead of trying all subsets of $\mathcal C$, we use Lemma 
	\ref{lemma: number minimal k-attractors}
	and generate only subsets of $\mathcal C$ of size at most $\sigma^k$; these subsets will include all minimal $k$-attractors and, in particular, all minimum $k$-attractors.
	By Lemma
	\ref{lemma: number minimal k-attractors},
	there are at most $exp\left(\bigO(\sigma^k\log\sigma^k)\right)$ such sets, and each verification takes linear $\bigO(|C|) \leq exp(\bigO(\log\sigma^k))$ time using Theorem \ref{th:checking}. Overall, our algorithm for the minimum $k$-attractor runs in 
	$\bigO(n) +  exp(\bigO(\log\sigma^k))\cdot exp\left(\bigO(\sigma^k\log\sigma^k)\right)$
	time, which simplifies to $\bigO(n) + exp\big(\bigO(\sigma^k\log\sigma^k)\big)$.
\end{proof}

\begin{proof}[Proof of Theorem \ref{theorem:find minimal}]
	We build the graph $\mathcal G_{S,k}$ of Definition \ref{def:marking graph} with the time/space bounds of Corollary \ref{cor:final G size}.
	
	We start with an empty $k$-attractor $\Gamma = \emptyset$  and with a set $A$ of \emph{active} vertexes initially defined as $A = \mathcal C$. Intuitively, $A$ will contain positions that can be removed from the $k$-attractor.
	
	Our algorithm to generate a minimal $k$-attractor works as follows:
	
	\begin{enumerate}
		\item Scan $\mathcal U$ and, for every $e\in \mathcal U$ having only one adjacent edge $\langle j,e \rangle$, remove $j$ from $A$ and insert it in $\Gamma$.
		\item Pick any element $j\in A$ and: 
		\begin{enumerate}
			\item For every $\langle j, e \rangle \in E$, if $e$ has only one other adjacent edge $\langle j',e \rangle$ such that $j\neq j'$, update $A \leftarrow A -\{j'\}$ and $\Gamma \leftarrow \Gamma \cup \{j'\}$.
			\item remove $j$ from $A$ and from $\mathcal G_{S,k}$, together with all its adjacent vertexes.
		\end{enumerate}
		\item If $A\neq \emptyset$, repeat from step (2). Otherwise, return $\Gamma$.
	\end{enumerate}
	
	\textbf{Correctness} At step (1) we find all $k$-necessary positions with respect to $A \cup \Gamma$ (Definition \ref{def:k-necessary}) and move them from $A$ to $\Gamma$. From this point, we maintain the following invariants: 
	
	\begin{itemize}
		\item (Inv 1) if $j\in A$, then $j$ is not $k$-necessary with respect to $A \cup \Gamma$, 
		\item (Inv 2) if $j\in \Gamma$, then $j$ is $k$-necessary with respect to $A \cup \Gamma$,  
		\item (Inv 3) $A\cap \Gamma = \emptyset$, and
		\item (Inv 4). $j \in A \cup \Gamma$ marks $e$ iff $\langle j,e \rangle$ is an edge of $\mathcal G_{S,k}$.
	\end{itemize}

	We show that the invariants hold after step (1). (Inv 1): assume, by contradiction, that there is a $j\in A$ that is $k$-necessary with respect to $A \cup \Gamma$, and note that $A \cup \Gamma = \mathcal C$. Then, by definition, there exists $e\in \mathcal U$ such that only $j$ marks $e$, i.e. $e$ has only one adjacent edge: $\langle j,e \rangle$. Then, step (1) would have removed $j$ from $A$, which yields a contradiction. (Inv 2): let $j\in\Gamma$. Then, in step (1) we have inserted $j$ in $\Gamma$ because there exists an $e\in \mathcal U$ having only one adjacent edge: $\langle j,e \rangle$. But then, $j$ is $k$-necessary w.r.t. $A \cup \Gamma$. (Inv 3): trivial, since at the beginning $\Gamma=\emptyset$ and we only move some elements from $A$ to $\Gamma$. (Inv 4) Trivial by definition of $\mathcal G_{S,k}$, since $A\cup \Gamma$ is unchanged by step (1).
	
	We now show that step (2) preserves the validity of the invariants.
	Let $\hat A$, $\hat \Gamma$, and $\hat{\mathcal G}_{S,k}$ denote $A$, $\Gamma$, and $\mathcal G_{S,k}$ \emph{before} entering in step (2), and $A$, $\Gamma$, and $\mathcal G_{S,k}$ denote the same objects upon exiting step (2).
	
	\ \\(Inv 1). Let $j'\in A$, and assume by contradiction that $j'$ is $k$-necessary with respect to $A \cup \Gamma$. Then, this means that there exists $e \in \mathcal U$ marked only by $j'$ among the elements in $A \cup \Gamma$. Then, $j'\in \hat A$, since we did not add any element to $\hat A$. But then, by (Inv 1) $j'$ is not $k$-necessary with respect to $\hat A\cup \hat \Gamma$; then, this means that there is a $j\in \hat A\cup \hat \Gamma$, $j\neq j'$ marking $e$. This $j$ must be the $j$ removed at point (2.b), since it is not present in $A$ (and we only remove one element from $\hat A$, i.e. $j$). This means that $\langle j,e \rangle$  and $\langle j',e \rangle$ are the only two edges adjacent to $e$ in $\hat{\mathcal G}_{S,k}$. But then, $e$ passes the test at point (2.a), therefore $A = \hat A - \{j'\}$, i.e. $j'\notin A$. This is a contradiction.
	
	\ \\(Inv 2). Let $j'\in\Gamma$, and assume by contradiction that $j'$ is not $k$-necessary with respect to $A\cup \Gamma$. 
	Since $j'\in\Gamma$, then by (Inv 3) it is either in $\hat \Gamma$ or $\hat A$. 
	
	(i) If $j'\in\hat \Gamma$, by (Inv 2) it is $k$-necessary with respect to $\hat A \cup \hat \Gamma$. Let $e\in \mathcal U$ be marked by $j'$. Since $j'$ is not $k$-necessary with respect to $A\cup \Gamma$, there is a $j''\in A\cup \Gamma$, with $j''\neq j'$, that marks $e$. But then, note that $j''\in A\cup \Gamma \subseteq \hat A \cup \hat \Gamma$. This work for every edge $e$ marked by $j'$, which implies that $j'$ is not $k$-necessary with respect to $\hat A \cup \hat \Gamma$. This is a contradiction. 
	
	(ii) Let therefore $j'\in\hat A$. Since we also have that $j'\in\Gamma$, this means that either $j'\in\hat \Gamma$ (absurd by Inv 3), or that	we moved $j'$ from $\hat A$ to $\hat \Gamma$ in step (2.a). Then, note that we only move $j'$ from $\hat A$ to $\hat \Gamma$ in step (2.a) when there is a $e\in\mathcal U$ adjacent only to $j'$ and some $j\in \hat A$, with $j\neq j'$. 
	By (Inv 4), this means that $j',j$ are the only elements marking $e$ among the elements in $\hat A \cup \hat \Gamma$. This $j$ is then removed from $\hat A$ in step (2.b). This means that $j'$ is the only element marking $e$ among the elements in $A \cup \Gamma$, i.e. $j'$ is $k$-necessary with respect to $A \cup \Gamma$. This is, again, a contradiction.

	\ \\(Inv 3). Trivial, since we either move elements from $A$ to $\Gamma$ or remove elements from $A$.
	
	\ \\(Inv 4). Trivial, since whenever we remove an element $j$ from $A \cup \Gamma$ we also remove $j$ (with all its adjacent edges) from $\mathcal G_{S,k}$.
	
	\ \\Since at the end of the algorithm we have that $A = \emptyset$, (Inv 2) guarantees that every $j\in\Gamma$ is $k$-necessary with respect to $A\cup \Gamma = \Gamma$, i.e. $\Gamma$ is $k$-necessary. We now prove that $A\cup \Gamma = \Gamma$ is a $k$ attractor at any point of the algorithm.
	
	At the beginning, $A\cup \Gamma = \mathcal C$ is a $k$-attractor (the complete one). Then, note that we remove positions $j$ from $A\cup \Gamma$ (2.b) only when $j\in A$. By (Inv 1), such a position is not $k$-necessary, i.e. each edge that it marks is also marked by some other $j'\neq j$. It follows that removing $j$ does not leave any edge unmarked. This proves that, at the end of the algorithm, $A\cup \Gamma = \Gamma$ is a $k$-attractor. Since $\Gamma$ is also $k$-necessary, by Theorem \ref{theorem:minimal characterization} $\Gamma$  is a minimal $k$-attractor.
	
	\ \\\textbf{Complexity} Step 1 takes $\bigO(|\mathcal U|)$ time. At step 2.a, note that we visit the following edges: (i) those adjacent to $j\in A$, plus (ii) edges $\langle j', e\rangle$ such that $e$ is adjacent to $j$ and $e$ has only two adjacent vertexes. Since we later remove $j$ and $j'$ from $A$, those edges will not be visited anymore. It follows that overall we spend $\bigO(|E|)$ time inside step 2. Overall, our algorithm runs in an expected time linear in the graph's size. 
\end{proof}

\begin{proof}[Proof of Theorem \ref{theorem:minimal approx rate}]
	Let $\Gamma$ be a minimal attractor. Let $\mathcal T_k = ST^k(S)$, and let $\mathcal L\subseteq \mathcal U$  be the set of the edges adjacent to the leaves of $\mathcal T_k$.
	
	
	By Theorem \ref{theorem:minimal characterization}, each $j\in\Gamma$ is $k$-necessary. By Definition \ref{def:k-necessary} this means that, for every $j\in\Gamma$, there exists \emph{at least} one edge $e\in \mathcal U$ marked only by $j$.

	Let $edge:\Gamma \rightarrow \mathcal U$ be the function defined in the proof of Lemma \ref{lemma: number minimal k-attractors}: this function maps each element $j\in \Gamma$ to the unique edge (at string depth at most $k$) $e=edge(j)$ such that (i) $e$ is marked  by $j$ only, and (ii) among all edges marked  by $j$ only, $e$ is the one with the lexicographically smallest $s(e)$.

	Let now $ledge:\Gamma \rightarrow \mathcal L$ (\emph{ledge} stands for \emph{leaf-edge}) be the function defined as follows. Let $e = edge(j)$. Then, $j$ is straddled by some occurrence of the substring $s(e)$. Among those occurrences, consider the one $S[i..i+|s(e)|-1] = s(e)$ (with $i \leq j \leq i+|s(e)|-1$) such that $j-i$ is maximized. Then, $S[i..i+k-1]$ (or $S[i..n]$ if $i+k-1 > n$) defines a unique path starting at the root of $\mathcal T_k$ and ending in the label of some edge $\ell\in\mathcal L$. By definition this $\ell$ is unique, and we define $ledge(j) = \ell$. We now show that $ledge$ is an injective function. 
	
	Let $j,j'\in\Gamma$ such that $ledge(j) = ledge(j') = \ell$ and assume, by contradiction, that $j\neq j'$. Let $e=edge(j)$ and $e'=edge(j')$. By definition of $ledge$, both $e$ and $e'$ lie on the path starting from the root of $\mathcal T_k$ and ending in $\ell$. Assume, without loss of generality,  that $|s(e')|>|s(e)|$, i.e. that $e$ is closer to the root than $e'$. Then, by definition of $ledge$, the occurrence $S[i..i+|s(e)|-1] = s(e)$ such that $i \leq j \leq i+|s(e)|-1$ and $j-i$ is maximized prefixes an occurrence of $s(\ell)$, i.e. $S[i..i+|s(\ell)|-1] = s(\ell)$.
	In particular, since $e'$ is on the path from $e$ to $\ell$ then the occurrence $S[i..i+|s(e)|-1]$ prefixes also an occurrence of $s(e')$: $S[i..i+|s(e')|-1] = s(e')$.
	Moreover, $i \leq j \leq i+|s(e)|-1 \leq i+|s(e')|-1$, i.e. $j\neq j'$ marks $e'$. This  contradicts  the way we defined $e'=edge(j')$, since $edge(j')$ yields an edge marked  by $j'$ \emph{only}. Therefore, $ledge$ is injective, which yields $|\Gamma| \leq |\mathcal L|$.

	Now, consider a minimum $k$-attractor $\Gamma^*_k$ of size $\gamma^*_k$. $\Gamma^*_k$ has to mark, in particular, the $|\mathcal L|$ edges adjacent to the leaves of $\mathcal T_k$. Note that each $j\in \Gamma^*_k$ is straddled by at most $k$ distinct $k$-mers; since each of these $k$-mers defines exactly one root-to-leaf path in $\mathcal T_k$, then $j$ can mark at most $k$ edges adjacent to the leaves. But then, $\Gamma^*_k$ needs at least $\gamma^*_k \geq |\mathcal L|/k$ elements in order to mark all $|\mathcal L|$ edges adjacent to the leaves.
	
	Putting these two bounds together, we obtain $\frac{|\Gamma|}{\gamma^*_k} \leq \frac{|\mathcal L|}{|\mathcal L|/k} \leq k$.
\end{proof}

\begin{proof}[Proof of Lemma \ref{lemma:queue}]
	
	Let $x_1, \dots, x_q$ be the $q$ input elements.
	We connect elements with the same $p(x)$ in a doubly-linked list (at most $M$ linked lists). These linked-lists are implemented as follows.
	We keep a hash table $H$ of size $\bigO(q)$ such that:
	\begin{itemize}
		\item $H[x_j]=NULL$ if $x_j$ is no longer in the queue. Otherwise:
		\item A pointer $H[x_j].next$ to the  element following $x_j$ in its linked list ($H[x_j].next=NULL$ if $x_j$ is the last),
		\item A pointer $H[x_j].prev$ to the  element preceding $x_j$ in its linked list ($H[x_j].prev=NULL$ if $x_j$ is the first),
	\end{itemize}
	
	We keep an array $P[1..M]$. Each $P[i]$ stores either:
	\begin{itemize}
		\item $P[i]=NULL$ if there are no elements with $p(x)=i$ in the queue, or
		\item $P[i] = x$ such that $x$ is the first element in the linked list associated with priority $p(x)$.
	\end{itemize}
	
	We moreover keep a variable $max$ storing the maximum index in $[1..M]$ such that $P[max]\neq NULL$. At the beginning, $max = \max\{p(x_1), \dots, p(x_q)\}$. All these structures can easily be initialized in $\bigO(M+q)$ time. 
	
	$\mathcal Q.pop()$ is implemented as follows. Let $x = P[max]$. If $H[x].next \neq NULL$, we update $P[max] \leftarrow H[x].next$. If $H[x].next = NULL$, we find the maximum $max'<max$ such that $P[max']\neq NULL$ and update $max \leftarrow max'$. $max'$ is found with a linear scan from $max-1$. If such a $max'$ does not exists, then the queue is empty and no operations can be done anymore. Then, we remove $x$ from the queue with the operation $H[x]\leftarrow NULL$ and return it.
	
	$\mathcal Q.dec(x)$ is implemented as follows. 
	If $p(x)=max$ and $x$ is the only element in its list, then we update $max \leftarrow max-1$.
	We remove $x$ from its  list and insert it at the head of the list with head $P[p(x)-1]$ (this list could be empty; in this case,  a new list with head $x$ is created at $P[p(x)-1]$).
	Then, we update $p(x)\leftarrow p(x)-1$. 
	
	Note that all operations take expected constant time, except finding the maximum $max'<max$ while computing $\mathcal Q.pop()$. However, $max$ always decreases, so overall we do not spend more than $\bigO(M)$ time in this step. Our claim easily follows.
	
\end{proof}

\begin{proof}[Proof of Theorem \ref{theorem:approx}]
	
	All we need to do is to find a subset $\Gamma$ of our covering set $\mathcal C$ by repeatedly inserting in $\Gamma$ the element of $\mathcal C$ that marks the largest number of unmarked edges in $\mathcal U$ (starting with $\Gamma=\emptyset$ and stopping as soon as all elements of $\mathcal U$ are marked). Our algorithm works as follows. 
	
	First, we build the graph $\mathcal G_{S,k}$ of Definition \ref{def:marking graph} with the time/space bounds of Corollary \ref{cor:final G size}. We moreover initialize the empty queue $\mathcal Q$ of Lemma \ref{lemma:queue} over the universe $\mathcal C$ and with priority function $p:\mathcal C \rightarrow [1..|\mathcal U|]$ initially defined as $p(x)=deg(x)$ for each $x\in \mathcal C$. Let $\Gamma \leftarrow \emptyset$ be our starting $k$-attractor. 
	Let moreover $M \leftarrow \emptyset$ be a set keeping track of the \emph{marked} suffix tree edges, and let $\mathcal G_{S,k}(\mathcal U)$ denote the elements from the original set $\mathcal U$ that are still vertexes of $\mathcal G_{S,k}$ (in the algorithm we will prune the set $\mathcal G_{S,k}(\mathcal U)$). At the beginning, $\mathcal G_{S,k}(\mathcal U) = \mathcal U$
	Then:
	
	\begin{enumerate}
		\item For each $x\in \mathcal C$,  call $\mathcal Q.insert(x)$.
		\item Pick $y \leftarrow \mathcal Q.pop()$ and insert $y$ in $\Gamma$. For each $\langle y,e \rangle\in E$: 
		\begin{enumerate}
			\item for each $\langle y', e\rangle \in E$,  call $\mathcal Q.dec(y')$.
			\item  remove $e$ and all its adjacent edges from from $\mathcal G_{S,k}$.
			\item Update $M \leftarrow M \cup \{e\}$.
		\end{enumerate}
		\item If $\mathcal G_{S,k}(\mathcal U) \neq \emptyset$,  repeat from step (2). Otherwise, return $\Gamma$.
	\end{enumerate}

	\textbf{Correctness} It is sufficient to show that the following invariants are valid after step (1) and are maintained by step (2): 
	
	\begin{itemize}
		\item (Inv 1) For every $e\in M$, there is a $j\in \Gamma$ that marks $e$,
		\item (Inv 2) For every $e\in \mathcal G_{S,k}(\mathcal U)$, none of the elements in $\Gamma$ marks $e$,
		\item (Inv 3) For every $x\in\mathcal Q$, $p(x)$ is precisely the number of elements in $\mathcal G_{S,k}(\mathcal U)$ that are marked by $x$, and 
		\item (Inv 4) $\mathcal G_{S,k}(\mathcal U) \cup M = \mathcal U$ and $\mathcal G_{S,k}(\mathcal U) \cap M = \emptyset$.
	\end{itemize}
	
	(Inv 1) is trivially valid after step (1) since at the beginning $M = \emptyset$. (Inv 2) is trivially valid after step (1) since at the beginning $\Gamma = \emptyset$. 
	(Inv 3) is clearly valid after step (1) since we define $p(x) = deg(x)$, and $deg(x)$ is the number of elements in $\mathcal U = \mathcal G_{S,k}(\mathcal U)$ adjacent to $x$ in the graph $\mathcal G_{S,k}$, i.e. the number of suffix tree edges marked by $x$ (by definition of $\mathcal G_{S,k}$). (Inv 4) is also clearly valid after step (1) by the way we initialize $\mathcal G_{S,k}(\mathcal U)$ and $M$. We now show that step (2) preserves the invariants.\\\ \\
	(Inv 1) for every new element $e$ inserted in $M$ at step (2.c), we also insert an element $y$ in $\Gamma$ such that $\langle y,e \rangle\in E$, i.e. $y$ marks $e$. By inductive hypothesis, the remaining elements in $M$ were already marked before entering in step (2), therefore the invariant is preserved.\\\ \\
	(Inv 2) Easy, since whenever we insert an element $y$ in $\Gamma$, we also remove from $\mathcal G_{S,k}(\mathcal U)$ all elements $e$ marked by $y$ (step 2.b).\\\ \\
	(Inv 3) note that, whenever we decrement $p(y')$ at step (2.a) with operation $\mathcal Q.dec(y')$, we also remove from $\mathcal G_{S,k}(\mathcal U)$ an element $e\in \mathcal G_{S,k}(\mathcal U)$ marked by $y'$ (step 2.b). 
	This is the only place where $p(y')$ is modified, therefore the validity of the invariant easily follows. \\\ \\
	(Inv 4) Easy, since whenever we remove an element $e\in \mathcal G_{S,k}(\mathcal U)$ (step 2.b) we also add it to $M$ (step 2.c).\\\ \\
	At the end of the procedure,  $\mathcal G_{S,k}(\mathcal U) = \emptyset$ and therefore, by (Inv 4), $M = \mathcal U$.	 
	By invariant (Inv 1), For every $e\in M = \mathcal U$ there is a $j\in \Gamma$ that marks $e$, i.e. $j$ is a $k$-attractor. Invariants (Inv 2) and (Inv 3) moreover guarantee that at each step we insert in $\Gamma$ the element $j$ marking the largest number of unmarked suffix tree edges, which implies that ours is the greedy algorithm achieving approximation rate $\mathcal H(k(k+1)/2)$.\\
	
	\textbf{Complexity} It is easy to see that we do not visit each edge of $\mathcal G_{S,k}$ more than once, and we perform a constant number of operations per edge. From Lemma \ref{lemma:queue}, our algorithm runs in expected linear time in the size of $\mathcal G_{S,k}$.
	
\end{proof}

\begin{proof}[Proof of Theorem \ref{th:sharp-attractor-npc}]
	We show a polynomial time reduction from \textsc{$k$-SetCover} to
	\textsc{$k$-Sharp-Attractor}. 
	$m$  denotes the number of sets in the collection.
	Denote the sizes of individual sets in
	the collection $C$ by $n_i=|C_i|$ and let
	$C_i=\{c_{i,1},c_{i,2},\ldots,c_{i,n_i}\}$. Let $
	\Sigma=\bigcup_{i=1}^{n}\bigcup_{j=1}^{k}\{x_i^{(j)}\} \cup
	\bigcup_{i=1}^{m}\bigcup_{j=1}^{n_i+1}\{\$_{i,j}\} \cup \{\#,\$\} $
	be our alphabet. Note that in the construction below $x_i^{(j)}$
	denotes a single symbol, while $\#^{k-1}$ denotes a concatenation of
	$k-1$ occurrences of $\#$.  We will now build a string $S$ over the
	alphabet $\Sigma$.  Let $ S = R\cdot \prod_{i=1}^{m}S_i $ where
	$\cdot$/$\prod$ denotes the concatenation of strings and $R$ and
	$S_i$ are defined below.
	
	Intuitively, we associate each $u\in \mathcal{U}$ with the substring
	$x_{u}^{(1)}\cdots x_{u}^{(k)}$ and each collection $C_i$ with
	substring $S_i$.  Each $S_i$ will contain all $n_i$ strings
	corresponding to elements in $C_i$ as substrings. The aim of $S_i$
	is to simulate --- via how many positions within $S_i$ are used in
	the solution to the \textsc{$k$-Sharp-Attractor} on $S$ --- the
	choice between not including $C_i$ in the solution to
	\textsc{$k$-Sharp-SetCover} (in which case $S_i$ is covered using a
	minimum possible number of positions that necessarily leaves
	uncovered all substrings corresponding to items in $C_i$) or
	including $C_i$ (in which case, by using only one additional
	position in the cover of $S_i$, the solution covers all substrings
	unique to $S_i$ \emph{and simultaneously} all $n_i$ substrings of
	$S_i$ corresponding to items in $C_i$). Gadget $R$ is used to cover
	``for free'' the substring $\#^k$ so that any algorithm solving
	\textsc{$k$-Sharp-Attractor} for $S$ will not have to optimize for
	its coverage at the boundaries of $S_i$.  This will be achieved as
	follows: $R$ will have $k+1$ length-$k$ substrings that appear only
	in $R$ and nowhere else in $S$. Thus, any $k$-sharp-attractor for
	$S$ has to include at least two positions within $R$. On the other
	hand, we will show that there exists a choice of two positions
	within $R$ that covers all those unique substrings, plus the
	substring $\#^k$ that we want to cover ``for free'' within $R$.
	
	For $i\in\{1,2,\ldots,m\}$ let (brackets added for clarity)
	\[
	S_i = \left( \prod_{j=1}^{n_i} \#^{k-1} \$_{i,1} \cdots \$_{i,j}
	x_{c_{i,j}}^{(1)} \cdots x_{c_{i,j}}^{(k)} \$_{i,j} \right) \cdot
	\#^{k-1} \$_{i,1} \cdots \$_{i,n_i+1} \#^{k-1}
	\]
	
	The example of $S_i$ for $k=6$ and $n_i=4$ is reported below. The
	meaning of overlined and underlined characters is explained below
	the example.
	\begin{align*}
		S_i =\ & \# \# \# \# \# \underline{\$_{i,1}}
		\overline{x_{c_{i,1}}^{(1)}} x_{c_{i,1}}^{(2)} x_{c_{i,1}}^{(3)}
		x_{c_{i,1}}^{(4)} x_{c_{i,1}}^{(5)} x_{c_{i,1}}^{(6)}
		\overline{\underline{\$_{i,1}}} \# \# \# \# \# \$_{i,1}
		\underline{\$_{i,2}} \overline{x_{c_{i,2}}^{(1)}}
		x_{c_{i,2}}^{(2)} x_{c_{i,2}}^{(3)} x_{c_{i,2}}^{(4)}
		x_{c_{i,2}}^{(5)} x_{c_{i,2}}^{(6)}
		\\ \ \ &\overline{\underline{\$_{i,2}}} \# \# \# \# \# \$_{i,1}
		\$_{i,2} \underline{\$_{i,3}} \overline{x_{c_{i,3}}^{(1)}}
		x_{c_{i,3}}^{(2)} x_{c_{i,3}}^{(3)} x_{c_{i,3}}^{(4)}
		x_{c_{i,3}}^{(5)} x_{c_{i,3}}^{(6)}
		\overline{\underline{\$_{i,3}}} \# \# \# \# \# \$_{i,1} \$_{i,2}
		\$_{i,3} \underline{\$_{i,4}} \overline{x_{c_{i,4}}^{(1)}}
		x_{c_{i,4}}^{(2)} \\ \ \ & x_{c_{i,4}}^{(3)}x_{c_{i,4}}^{(4)}
		x_{c_{i,4}}^{(5)} x_{c_{i,4}}^{(6)}
		\overline{\underline{\$_{i,4}}} \# \# \# \# \# \overline{\$_{i,1}}
		\$_{i,2} \$_{i,3} \$_{i,4} \overline{\underline{\$_{i,5}}} \# \#
		\# \# \#
	\end{align*}
	
	Any $k$-sharp-attractor of $S$ contains at least $2n_i+1$ positions
	within $S_i$ because $S_i$ contains $2n_i+1$ non-overlapping
	substrings of length $k$ such that each
	\emph{necessarily}\footnote{That is, independent of what is $C_i$.
		Importantly, although any of the substrings in
		$\bigcup_{j=1}^{k}\{x_{c_{i,j}}^{(1)}\cdots x_{c_{i,j}}^{(k)}\}$
		could have the only occurrence in $S_i$, they are not
		\emph{necessarily} unique to $S_i$. This situation is analogous to
		\textsc{$k$-SetCover} when some $u\in\mathcal{U}$ is covered by
		only one set in $C$, and thus that set has to be included in the
		solution.}  occurs in $S_i$ only once and nowhere
	else\footnote{This can be verified by consulting the definition of
		$R$ below.} in $S$: $\bigcup_{j=1}^{n_i}\{\$_{i,j}\#^{k-1}\} \cup
	\bigcup_{j=1}^{n_i+1}\{\$_{i,j}\#^{k-1}\}$.  With this in mind we
	now observe that $S_i$ has the following two properties:
	\begin{enumerate}
		\item There exists a ``minimum'' set $\Gamma_{S,i}$ of $2n_i+1$
		positions within the occurrence of $S_i$ in $S$ that covers all
		substrings of $S_i$ of length $k$ that necessarily occur only in
		$S_i$ and nowhere else in $S$. The set $\Gamma_{S,i}$ includes
		the first occurrence of $\$_{i,j}$ for $j\in\{1,\ldots,n_i+1\}$
		and the second occurrence of $\$_{i,j}$ for
		$j\in\{1,\ldots,n_i\}$ ($\Gamma_{S,i}$ is shown in the above
		example using underlined positions). Furthermore, it is easy to
		check that $\Gamma_{S,i}$ is \emph{the only} such set (see also
		the proof of~\cite[Thm. 8]{kempa2017roots}).  We now observe
		that the only substrings of $S_i$ of length $k$ not covered by
		$\Gamma_{S,i}$ are precisely strings
		$\bigcup_{j=1}^{k}\{x_{c_{i,j}}^{(1)}\cdots
		x_{c_{i,j}}^{(k)}\}$.  Thus, if in any $k$-sharp-attractor of
		$S$, $S_i$ is covered using the minimum number of $2n_i+1$
		positions, these positions \emph{must} be precisely
		$\Gamma_{S,i}$ and hence, in particular, any of the strings from
		the set $\bigcup_{j=1}^{k}\{x_{c_{i,j}}^{(1)}\cdots
		x_{c_{i,j}}^{(k)}\}$ is not covered within $S_i$.
		\item There exists a ``universal'' set $\Gamma_{S,i}'$ of $2n_i+2$
		positions within the occurrence of $S_i$ in $S$ that covers
		\emph{all} substrings of $S_i$ of length $k$.  The set
		$\Gamma_{S,i}'$ includes: the only occurrence of
		$x_{c_{i,j}}^{(1)}$ for $j\in\{1,\ldots,n_i\}$, the second
		occurrence of $\$_{i,j}$ for $j\in\{1,\ldots,n_i\}$, the only
		occurrence of $\$_{i,n_i+1}$, and the last occurrence of
		$\$_{i,1}$ ($\Gamma_{S,i}'$ is shown in the above example using
		overlined positions).
	\end{enumerate}
	
	To finish the construction we will ensure that the substring $\#^k$
	is covered ``for free'' in $S$. This will allow us to easily exclude
	$\#^k$ from consideration, therby simplifying the proof.  To this
	end we introduce a string $R=\#^k \$ \$ \#^{k-1}$.  Any
	$k$-sharp-attractor of $S$ contains at least two positions within
	$R$ because $R$ contains $k+1$ substrings of length $k$ that occur
	in $R$ and nowhere else in $S$.  On the other hand, the set
	$\Gamma_{R}=\{k,k+2\}$ of positions within $R$ covers all substrings
	of length $k$ occurring only in $R$ as well as the substring
	$\#^{k}$.
	
	With the above properties, we are now ready to prove our earlier
	claim: an instance $\langle \mathcal{U},C \rangle$ of
	\textsc{$k$-SetCover} has a solution of size $\leq p$ if and only if
	$S$ has a $k$-sharp-attractor of size $\leq 2t+m+p+2$, where
	$t=\sum_{i=1}^{m}n_i$.
	
	``$(\Rightarrow)$'' Let $C' \subseteq C$ be a cover of $\mathcal{U}$
	of size $p'\leq p$ and let $\Gamma_{C'}=
	\bigcup\{\Gamma_{S,i}'\ \mid C_i \in C'\} \cup
	\bigcup\{\Gamma_{S,i}\mid C_i\not\in C'\} \cup \Gamma_{R}$.  It is
	easy to check that $|\Gamma_{C'}|=2t+m+p'+2$.  From the discussion
	above $\Gamma$ covers all substrings of $S$ of length $k$ contained
	inside gadgets. In particular, $\Gamma_{C'}$ covers all strings
	$\{x_i^{(1)}\cdots x_i^{(k)}\}_{i=1}^{m}$ since $C'$ is a cover of
	$\mathcal{U}$ and thus the claim follows.
	
	``$(\Leftarrow)$'' Let $\Gamma$ be a $k$-sharp-attractor of $S$ of
	size $\leq 2t+m+p+2$. We will show that $\mathcal{U}$ must have a
	cover of size $\leq p$ using elements from $C$.  Let $\mathcal{I}$
	be the set of indices $i\in\{1,\ldots,m\}$ for which $\Gamma$
	contains more than $2n_i+1$ positions within the occurrence of $S_i$
	in $S$. By the above discussion, $\Gamma$ cannot have less than two
	positions within the occurrence of $R$ in $S$. Thus, there are at
	most $2t+m+p$ positions left to use within $\{S_i\}_{i=1}^{m}$.
	Each of $S_i$, $i\in\{1,\ldots,m\}$ requires $2n_i+1$ positions, and
	hence there cannot be more than $p$ indices where $\Gamma$ uses more
	positions than necessary. Thus, $|\mathcal{I}|\leq p$. Let
	$C'_{\Gamma}=\{C_i \in C \mid i \in \mathcal{I}\}$. We now show that
	$C'_{\Gamma}$ is a cover of $\mathcal{U}$. Suppose that there exists
	$u\in\mathcal{U}$ that is not in $\bigcup C'_{\Gamma}$. Consider
	then the set $\mathcal{I}_u=\{i\in\{1,\ldots,m\} \mid u\in
	C_i\}$. We must have $\mathcal{I}\cap \mathcal{I}_u=\emptyset$ and
	thus each $S_i$ corresponding to $C_i$ containing $u$ is covered in
	$\Gamma$ using the minimum attractor of size $2n_i+1$.  This
	implies, by the above discussion, that the string $x_u^{(1)}\cdots
	x_u^{(k)}$ is not covered by $\Gamma$, a contradiction.
\end{proof}

\end{document}